\documentclass[11pt,a4paper]{article}
\pdfoutput=1
\usepackage[utf8]{inputenc}
\usepackage{jheppub}
\hypersetup{unicode=true,bookmarksopen=true}

\usepackage{etoolbox}
\makeatletter
\patchcmd{\maketitle}{\@fpheader}{}{}{}
\makeatother

\usepackage{amsmath}
\usepackage{amssymb}
\usepackage{mathtools}
\usepackage{amsthm}
\usepackage{dsfont}

\usepackage{microtype}
\usepackage{lmodern}
\usepackage[T1]{fontenc}
\usepackage{bbm}
\DeclareMathAlphabet{\mathbfi}{OML}{cmm}{b}{it}

\usepackage{mathrsfs}

\let\originalleft\left
\let\originalright\right
\renewcommand{\left}{\mathopen{}\mathclose\bgroup\originalleft}
\renewcommand{\right}{\aftergroup\egroup\originalright}

\makeatletter
\newenvironment{equations}[1][]{\subequations\ifx\relax#1\relax\else\label{#1}\fi\align\ignorespaces}{\endalign\ignorespacesafterend\endsubequations}
\def\@spliteq#1{\begin{equation}\begin{split}#1\end{split}\end{equation}}
\def\@spliteqstar#1{\begin{equation*}\begin{split}#1\end{split}\end{equation*}}
\def\splitequation{\collect@body\@spliteq}
\expandafter\def\csname splitequation*\endcsname{\collect@body\@spliteqstar}

\expandafter\def\csname endsplitequation*\endcsname{\ignorespacesafterend}
\makeatother

\renewcommand{\vec}[1]{{\ifnum9<1#1\mathbf{#1}\else\ifcat\noexpand#1\relax\boldsymbol{#1}\else\mathbfi{#1}\fi\fi}}
\newcommand{\mathe}{\mathrm{e}}
\newcommand{\mathi}{\mathrm{i}}
\newcommand{\total}{\mathop{}\!\mathrm{d}}
\newcommand{\abs}[1]{{\left\lvert{#1}\right\rvert}}
\newcommand{\1}{\mathbbm{1}}
\newcommand{\eqend}[1]{\,#1}
\newcommand{\bigo}[1]{\mathcal{O}\left({#1}\right)}
\newcommand{\hypergeom}[2]{\,{}_{#1}\mathrm{F}_{#2}}
\newcommand{\Estoch}{\mathop{\mathds{E}}}

\allowdisplaybreaks

\theoremstyle{plain}
\newtheorem{theorem}{Theorem}[section]
\newtheorem{lemma}[theorem]{Lemma}
\newtheorem{proposition}[theorem]{Proposition}
\newtheorem{corollary}[theorem]{Corollary}

\theoremstyle{remark}
\newtheorem{remark}[theorem]{Remark}

\newcommand{\Eq}{Eq.}
\newcommand{\Eqs}{Eqs.}

\begin{document}

\title{Measurements in stochastic gravity\\ and thermal variance}

\author[a]{Markus B. Fr{\"o}b,}
\author[b]{Dra\v{z}en Glavan,}
\author[c,d]{and Paolo Meda}

\affiliation[a]{Institut f{\"u}r Theoretische Physik, Universit{\"a}t Leipzig, Br{\"u}derstra{\ss}e 16, 04103 Leipzig, Germany}
\affiliation[b]{CEICO, FZU --- Institute of Physics of the Czech Academy of Sciences, Na Slovance 1999/2, 182 21 Prague 8, Czech Republic}
\affiliation[c]{Dipartamento di Matematica, Università di Trento, Via Sommarive 14, 38123 Povo (Trento), Italy}
\affiliation[d]{Trento Institute for Fundamental Physics and Applications (TIFPA-INFN), Via Sommarive 14, 38123 Povo (Trento), Italy}

\emailAdd{mfroeb@itp.uni-leipzig.de}
\emailAdd{glavan@fzu.cz}
\emailAdd{paolo.meda@unitn.it}

\abstract{We analyze the thermal fluctuations of a free, conformally invariant, Maxwell quantum field (photon) interacting with a cosmological background spacetime, in the framework of quantum field theory in curved spacetimes and semiclassical and stochastic gravity. The thermal fluctuations give rise to backreaction effects upon the spacetime geometry, which are incorporated in the semiclassical Einstein--Langevin equation, evaluated in the cosmological Friedmann--Lema\^{i}tre--Robertson--Walker spacetime. We first evaluate the semiclassical Einstein equation for the background geometry sourced by the thermal quantum stress-energy tensor. For large enough temperature, the solution is approximated by a radiation-dominated expanding universe driven by the thermal bath of photons. We then evaluate the thermal noise kernel associated to the quantum fluctuations of the photon field using point-splitting regularization methods, and give its explicit analytic form in the limits of large and small temperature, as well as a local approximation. Finally, we prove that this thermal noise kernel corresponds exactly to the thermal variance of the induced fluctuations of the linearized metric perturbation in the local and covariant measurement scheme defined by Fewster and Verch. Our analysis allows to quantify the extent to which quantum fluctuations may give rise to non-classical effects, and thus become relevant in inflationary cosmology.}

\dedicated{Dedicated to Enric Verdaguer on the occasion of his 75th birthday.}

\maketitle

\section{Introduction}

The $\Lambda$CDM model is commonly held as the most convincing and observationally supported model in cosmology, with the initial inflationary phase of our universe extremely well described as driven by a single scalar field (the inflaton) with a specific potential. On top of the classical homogeneous background, small-scale quantum fluctuations which get amplified and classicalized during the inflationary phase are the seeds of structure formation at large scales~\cite{Starobinsky1982,Liddle2000}. Like all matter, also the quantum fluctuations must gravitate, which poses the question of what is their influence on the background, namely the backreaction upon the spacetime geometry~\cite{Wald1977back,Wald1995qftcurved,Mukhanov1997}.

In the semiclassical approach to gravity, backreaction effects are usually taken into account by studying the semiclassical Einstein equation, in which the classical Einstein tensor of the cosmological spacetime (e.g., Friedmann--Lema\^{i}tre--Robertson--Walker spacetimes) is equal to the expectation value of the quantum stress-energy tensor associated to the quantum field, evaluated in a physical Hadamard state. This equation models the quantum matter-energy content which interacts with gravity, and its self-consistent solutions formed by the spacetime metric and a quantum state for this metric encode the evolution of the matter-gravity system beyond the classical relativistic gravitational theory \cite{Hu1977,Hartle1979,Hack2016cosm,Pinamonti2011,Pinamonti2013,Gottschalk2018,Meda2020,Gottschalk2023}.

A modification of the semiclassical Einstein equation which extends the semiclassical description of the interplay between quantum matter and gravity is provided by the Einstein--Langevin equation, which apart from the expectation value of the quantum stress-energy tensor also incorporate its fluctuations through a stochastic source or noise \cite{Hu1994,Campos1994,Campos1995,Calzetta1997,Martin1998,Martin1998_2,Martin1999,Hu2002} (see also \cite{Hu2008,Hu2020stoc,Eftekharzadeh2012,Satin2015,Satin2016} and references therein). This source is Gaussian with vanishing expectation value, but non-vanishing variance given by the two-point function of the quantum stress-energy tensor, termed the noise kernel. In turn, the stochastic source induces fluctuations of the metric, which thus also behaves stochastically. It is then possible to formulate criteria for the validity of semiclassical gravity, and thus to establish to what extent quantum metric fluctuations become important in gravitational backgrounds \cite{Anderson2003,Hu2000,Hu2001,Hu2004,Huvalid2004}.

Backreaction effects due to quantum matter and its fluctuations have been studied extensively, and for a partial list of references about the connection of this topic with stochastic gravity we refer the reader to \cite{Calzetta1995,Calzetta1997,Roura1999,Roura2007,Urakawa2007,Urakawa2008,Hu2020stoc} and references therein. Very recently, it has been shown \cite{Ota2023,Froeb2025} that the backreaction of thermal fluctuations of quantum photons in the radiation-dominated era can lead to secular effects in the power spectra of primordial tensor perturbations in the superhorizon limit of small momentum $\abs{\vec{k}} \ll H a$. In particular, the noise kernel can be approximated by a local term corresponding to an effective tachyonic mass in the dynamical equation for the tensor perturbations. Another interesting effect concerns stimulated emission, which also can be enhanced by backreaction in a thermal bath~\cite{Ota2024,Ota2025}.

\medskip

Our analysis sits in this semiclassical and stochastic framework, and it aims to evaluate the noise kernel associated to a thermal bath of free, conformally invariant, Maxwell quantum fields (photons) in cosmological spacetime. To model the thermal bath, we select a conformally thermal quasi-free state, which reduces to the conformal vacuum state in the limit of vanishing temperature. This state satisfies the Hadamard condition, exists for arbitrary Friedmann-Lema\^{i}tre-Robertson-Walker (FLRW) spacetimes, and may be obtained via a bulk-to-boundary correspondence in these spacetimes. In this approach, both the Hadamard condition and the thermal property of a distinguished quantum state on the null boundary of the spacetime can be pulled-back into the bulk, and for conformally invariant quantum fields such as Maxwell fields, the resulting quantum state is a conformal KMS state on the full spacetime with respect to the boundary time translation \cite{Dappiaggi2007,Dappiaggi2008,Dappiaggi2010,Dappiaggi2017}. For references about Hadamard states of Maxwell fields, see also \cite{Finster2013,Siemssen2013,Murro2024}.

For these conformally thermal states, we study the stochastic Einstein--Langevin equation, in which the full metric is split into the cosmological FLRW background metric and a linearized perturbation, which is in turn composed of intrinsic and induced perturbations. The source of the intrinsic perturbations is the quantum stress-energy tensor evaluated in that conformally thermal state, whereas the source of the induced perturbations is a stochastic source, whose variance encodes the quantum fluctuations of the quantum matter field as explained above. At the background level, we evaluate the expectation value of the quantum stress-energy tensor of the Maxwell field exploiting the Hadamard property of conformally thermal states, such that the difference between that expectation value and the one obtained in the conformal vacuum is a smooth function. Thus, using the quantum trace anomaly to obtain the expectation value of $T_{\mu\nu}$ in the conformal vacuum, we obtain the explicit form of the thermal quantum stress-energy tensor which sources the zeroth-order semiclassical Einstein equation. We show that for large enough temperature, the unique solution of this equation can be well approximated by a radiation-dominated universe driven by the thermal bath of photons. At the linear level, we evaluate the noise kernel which fully characterizes the stochastic source in the Einstein--Langevin equation, and derive explicit expressions for the noise kernel in some physically reasonable approximations: $\beta \to 0$ (large temperature), $\beta \to \infty$ (small temperature), and $\abs{\vec{p}} \to 0$ (small external momentum).  

Finally, we study the nature of the thermal noise kernel in the framework of local and covariant measurement scheme defined by Fewster and Verch in \cite{Fewster2018} (see also \cite{Fewster2022,Fewster2025}). In this approach, we evaluate the quantum observable induced by the coupling between the stochastic noise and the linear metric perturbation, which play the roles of target field and probe field, respectively. To model the actual measurement, we prepare a combined state at early times, i.e., before the measurement process, which is the tensor product between the quasi-free conformally thermal state characterizing the stochastic noise, and a quasi-free reference state for the metric perturbation.

We then prove that the variance of the actual measurement in the interacting combined state is the sum of the variance of the intrinsic perturbation, evaluated in the quasi-free reference state, and of the variance of the induced perturbation, evaluated in the conformally thermal state. In particular, the latter contribution is obtained by applying the causal propagator of the linearized metric perturbation to each argument of the thermal noise kernel, and smearing it with the test tensor of the initial probe observable. Namely, the thermal noise kernel corresponds to the gauge-invariant integral kernel of the stochastic variance of the induced linearized metric perturbation, as expected in stochastic gravity. This last result determines a minimum measurement uncertainty for the induced metric perturbations driven by the thermal quantum fluctuations. Such an unavoidable limit in the accuracy of the measurement of quantum fluctuations provides an important constraint for future low-energy experiments of quantum gravity effects, both in flat space and in cosmology.

\medskip

\paragraph{Outline.} The paper is organized as follows: In Section \ref{sec:stochastic} we recall the definition of flat FLRW spacetime as a conformally flat cosmological spacetime, and we provide a brief review of stochastic gravity and the Einstein--Langevin equation, with particular emphasis on the role played by the intrinsic and induced linearized perturbations of the metric around a fixed background geometry. Furthermore, we introduce the stochastic Gaussian noise describing quantum fluctuations, which is completely characterized by a vanishing expectation value and by a non-vanishing variance defined by the noise kernel. Section \ref{sec:photon} recalls the properties of a free, thermal, conformally invariant quantum Maxwell field on cosmological spacetimes in the algebraic framework. It also summarizes the construction of the quantum stress-energy tensor of the Maxwell field in the Hadamard regularization scheme, which provides a well-defined trace anomaly in terms of some coefficients of the Hadamard parametrix. Based on this, the evaluation of the quantum stress-energy tensor in the conformal vacuum and in the conformally thermal state is provided, and the solution for the zeroth-order semiclassical Einstein equation driven by the quantum field is given. Section \ref{sec:thermal-noise} is devoted to the exact evaluation of the thermal noise kernel, whose components are explicitly computed in various limits (large temperature, small temperature, and small external momentum). Section \ref{sec:measurement} contains the proof of the thermal variance of the induced fluctuations evaluated in the local and covariant measurement framework of Fewster and Verch, thus showing precisely the link between quantum fluctuations, thermal noise kernel, and induced metric perturbations. The last section contains the conclusions on the work and possible future developments. The technical evaluations and lemmas employed to compute the thermal noise kernel are collected in the appendix.

\section{Stochastic gravity}
\label{sec:stochastic}

\subsection{Cosmological spacetimes}

Our current understanding of the universe in cosmology is founded on the $\Lambda$CDM model, according to which it may be described by a flat Friedmann-Lema\^{i}tre-Robertson-Walker (FLRW) spacetime $(\mathcal{M},g)$ where $\mathcal{M} = I_t \times \mathbb{R}^3$ and $I_t \subset \mathbb{R}$ is an interval of time. In Cartesian cooordinates $x^\mu = (t,\vec{x})^\mu$ wtih cosmological time $t$, the metric of the spacetime is given by
\begin{equation}
g_{\mu\nu} \total x^\mu \total x^\nu = - \total t^2 + a(t)^2 \total \vec{x}^2 \eqend{,}
\end{equation}
where the scale factor $a \in C^n(I_t)$, with $n \geq 2$, constitutes the unique degree of freedom of the spacetime. Every FLRW spacetime is conformally flat, which is made manifest by employing the conformal coordinates $(\eta,\vec{x})$, where $\eta$ denotes the conformal time. In these coordinates, $a$ plays the role of conformal factor and the FLRW metric reads
\begin{equation}
\label{eq:FLRW}
g_{\mu\nu} = a(\eta)^2 \eta_{\mu\nu} \eqend{.}
\end{equation}
In the following and throughout the paper, we denote with $x^\mu = (\eta,\vec{x})$ and $(x')^\mu = (\eta',\vec{x}')$ the conformal coordinates.

From the scale factor, we define the functions
\begin{equation}
H(\eta) = \frac{a'}{a^2} \eqend{,} \quad \mathcal{R}(\eta) = \frac{1}{H a} \eqend{,} \quad \epsilon(\eta) = - \frac{H'}{H^2 a} \eqend{,}
\end{equation}
which denote the Hubble function, the comoving Hubble radius, and the dimensionless first slow-roll parameter, respectively. In the ``slow-roll'' approximation that describes well the primordial inflation, it is assumed that $\epsilon \ll 1$, but we will make no assumption about the size of $\epsilon$. For the background curvature tensors, we compute
\begin{equations}[eq:curv-FLRW]
R_{\mu\nu\rho\sigma} &= 2 H^2 a^4 \eta_{\mu[\rho} \eta_{\sigma]\nu} - 4 \epsilon H^2 a^4 \delta^0_{[\mu} \eta_{\nu][\rho} \delta^0_{\sigma]} \eqend{,} \\
R_{\mu\nu} &= (3-\epsilon) H^2 a^2 \eta_{\mu\nu} + 2 \epsilon H^2 a^2 \delta^0_\mu \delta^0_\nu \eqend{,} \\
R &= 6 (2-\epsilon) H^2 \eqend{,}
\end{equations}
and it follows that $C_{\mu\nu\rho\sigma} = R_{\mu\nu\rho\sigma} - R_{\mu[\rho} g_{\sigma]\nu} + R_{\nu[\rho} g_{\sigma]\mu} + \frac{1}{3} R g_{\mu[\rho} g_{\sigma]\nu} = 0$ as required for a conformally flat spacetime.

In order to study later the semiclassical and stochastic equations, we also evaluate the following two covariantly conserved curvature tensors of order 4. Denoting with $\total\mu_x \doteq \sqrt{-g(x)} \total^4 x$ the covariant measure on $\mathcal{M}$, we define
\begin{equations}
\label{eq:amunu_bmunu}
A^{\mu\nu} &\doteq \frac{1}{\sqrt{-g}} \frac{\delta}{\delta g_{\mu\nu}} \int C_{\alpha\beta\gamma\delta} C^{\alpha\beta\gamma\delta} \total\mu_x = - 2 \left( 2 \nabla_\rho \nabla_\sigma + R_{\rho\sigma} \right) C^{\mu\rho\nu\sigma} \eqend{,} \\
B^{\mu\nu} &\doteq \frac{1}{\sqrt{-g}} \frac{\delta}{\delta g_{\mu\nu}} \int R^2 \total\mu_x = \frac{1}{2} g^{\mu\nu} R^2 - 2 R R^{\mu\nu} + 2 \nabla^\mu \nabla^\nu R - 2 g^{\mu\nu} \nabla^2 R \eqend{,}
\end{equations}
satisfying $\nabla_\mu A^{\mu\nu} = \nabla_\mu B^{\mu\nu} = 0$. On cosmological spacetimes, the first tensor trivially vanishes because the Weyl tensor $C_{\mu\nu\rho\sigma}$ vanishes in every conformally flat spacetime such as FLRW \eqref{eq:FLRW}. On the contrary, the second tensor reads
\begin{splitequation}
\label{eq:bmunu-cosmo}
B_{\mu\nu} &= 6 \eta_{\mu\nu} H^2 \left[ - 2 \epsilon'' - 2 (5-6\epsilon) \epsilon' H a - 3 \epsilon  (2-\epsilon) (3-4\epsilon) H^2 a^2 \right] \\
&\quad+ 12 \delta^0_\mu \delta^0_\nu H^2 \left[ - \epsilon'' - 2 (1-3\epsilon) \epsilon' H a + 6 (2-\epsilon) \epsilon^2 H^2 a^2 \right] \eqend{.}
\end{splitequation}
In particular, its trace is given by
\begin{equation}
\label{eq:bmunu-trace}
g^{\mu\nu} B_{\mu\nu} = 6 \nabla^2 R = - 36 a^{-2} H^2 \left[ \epsilon'' + 6 (1-\epsilon) \epsilon' H a + 6 \epsilon (2-\epsilon) (1-\epsilon) H^2 a^2 \right] \eqend{.}
\end{equation}
\begin{remark}
A third contribution descending from $R_{\rho\sigma\gamma\delta} R^{\rho\sigma\gamma\delta} \sqrt{-g}$ which a priori should be also included in the list of covariantly conserved curvature ensors of order 4 can be removed, because the integral of the Euler density $\left( R^2 - 4 R_{\rho\sigma} R^{\rho\sigma} + R_{\rho\sigma\gamma\delta} R^{\rho\sigma\gamma\delta} \right) \sqrt{-g}$ is a topological invariant in four dimensions.
\end{remark}
\begin{remark}
An explicit solution for the scale factor $a(\eta)$ can be obtained for constant $\epsilon$, which is known as power-law inflation. In this case, we have
\begin{equation}
\label{eq:power-law}
H(\eta) = H_0 a^{-\epsilon}(\eta) \eqend{,} \quad a(\eta) = \left[ 1 - H_0 (1-\epsilon) (\eta-\eta_0) \right]^{- \frac{1}{1-\epsilon}} \eqend{,} \quad \mathcal{R}(\eta) = \frac{1}{1 - (1-\epsilon) (\eta-\eta_0)} \eqend{,}
\end{equation}
where $H_0 = H(\eta_0)$ and $\eta_0$ are constants. This also describes the periods of radiation domination ($\epsilon = 2$), matter domination ($\epsilon = \frac{3}{2}$) and the de Sitter solution ($\epsilon = 0$)~\cite{Liddle2000}.
\end{remark}

\subsection{Einstein--Langevin equation}

Stochastic gravity was first developed by Verdaguer and Hu as an extension of semiclassical gravity (for a list of references on this topic, see \cite{Hu2008,Hu2020stoc} and the references given in the Introduction). In this approach, the semiclassical approximation is extended by adding a Gaussian stochastic noise to the semiclassical Einstein equation. The noise contribution describes the fluctuations of the metric which are induced by the quantum matter and sourced by the two-point function of the quantum stress-energy tensor. To be more precise, one assumes that the full metric of the spacetime may be written as
\begin{equation}
\label{eq:metric-split}
g_{\mu\nu}(x) = g^{(0)}_{\mu\nu}(x) + \kappa h_{\mu\nu}(x) \eqend{,}
\end{equation}
where $g^{(0)}_{\mu\nu}$ denotes the background geometry, $h_{\mu\nu}$ denotes a linear perturbation of the background metric, and $\kappa^2 = 16 \pi G_N$ is the perturbative parameter depending on the Newton constant $G_N$ in the units $\hbar = c = 1$\footnote{From now on, we shall use the superscript $^{(0)}$ to denote the quantities fully constructed out of the background metric $g^{(0)}_{\mu\nu}$, and the superscript $^{(1)}$ to denote the quantities linearized in the perturbation $h_{\mu\nu}$.}. Then the Einstein--Langevin equation
\begin{splitequation}
\label{eq:Ein-Lang}
&G_{\mu\nu}[g](x) + \Lambda g_{\mu\nu}(x) - 2 \left[ \alpha_A A_{\mu\nu}[g](x) + \alpha_B B_{\mu\nu}[g](x) \right] \\
&\quad= \frac{1}{2} \kappa^2 \left[ \omega\left( T_{\mu\nu}[g](x) \right) + \xi_{\mu\nu}[g](x) \right]
\end{splitequation}
determines the backreaction of the quantum matter field, together with its quantum fluctuations, upon the spacetime geometry. Here, $A^{\mu\nu}$ and $B^{\mu\nu}$ were defined in \Eq~\eqref{eq:amunu_bmunu}, while $\alpha_A$ and $\alpha_B$ are arbitrary dimensionless constants. The derivation of this equation may be performed in different formalisms, e.g., in the large-$N$ expansion \cite{Hu2004} or using functional techniques~\cite{Martin1999}.

If one performs a perturbative expansion in $\kappa$ of the full metric in the background metric and the linear perturbation according to \Eq~\eqref{eq:metric-split}, then \Eq~\eqref{eq:Ein-Lang} may be rewritten as a pair of zeroth-order and linear order semiclassical equations for $g^{(0)}_{\mu\nu}$ and $h_{\mu\nu}$, namely
\begin{equations}
G^{(0)}_{\mu\nu}(x) + \Lambda g^{(0)}_{\mu\nu}(x) - 2 \left[ \alpha_A A^{(0)}_{\mu\nu}(x) + \alpha_B B^{(0)}_{\mu\nu}(x) \right] &= \frac{1}{2} \kappa^2 \omega\left( T^{(0)}_{\mu\nu}(x) \right) \eqend{,} \label{eq:Ein-Lang-zero} \\
G^{(1)}_{\mu\nu}(x) + \Lambda h_{\mu\nu}(x) - 2 \left[ \alpha_A A^{(1)}_{\mu\nu}(x) + \alpha_B B^{(1)}_{\mu\nu}(x) \right] &= \frac{1}{2} \kappa^2 \omega\left( T^{(1)}_{\mu\nu}(x) \right) + \frac{1}{2} \kappa \xi_{\mu\nu}(x) \label{eq:Ein-Lang-lin} \eqend{.}
\end{equations}
While the zeroth-order equation corresponds exactly to the semiclassical Einstein equation for the background metric $g^{(0)}_{\mu\nu}(x)$, the linearized equation depends on both the background metric and the linear perturbation, where $\xi_{\mu\nu}(x)$ plays the role of a classical source. Based on this, the linearized perturbation $h_{\mu\nu}(x)$ may be split into three contributions \cite{Roura1999}
\begin{equation}
\label{eq:pert-dec}
h_{\mu\nu}(x) = h_{\mu\nu}^\text{intr}(x) + h_{\mu\nu}^\text{ind}(x) \eqend{,} \qquad h_{\mu\nu}^\text{intr}(x) = h_{\mu\nu}^{\text{intr},0}(x) + h_{\mu\nu}^{\text{ind},1}(x) \eqend{.}
\end{equation}
On the one hand, $h_{\mu\nu}^{\text{intr}}$ describes the \textit{intrinsic} fluctuations of the metric, which are sourced by the expectation value of the quantum stress-energy tensor $\omega\left( T_{\mu\nu}(x) \right)$. The zeroth-order $\omega\left( T^{(0)}_{\mu\nu}(x) \right)$ is related to the free part $h_{\mu\nu}^{\text{intr},0}$ of the intrinsic fluctuations, which is thus part of the background solution, whereas the linearized order $\omega\left( T^{(1)}_{\mu\nu}(x) \right)$ sources the part $h_{\mu\nu}^{\text{intr},1}$ modulated by the interaction with the quantum matter. In the case of thermal photons, these one-loop contributions have been fully evaluated in \cite{Froeb2025} (see also \cite{Ota2023} where the one-loop contribution sourcing $h_{\mu\nu}^{\text{intr},1}$ was denoted radiation exchange). On the other hand, $h_{\mu\nu}^\text{ind}$ represents the \textit{induced} fluctuations of the metric sourced by the fluctuations of the quantum matter fields on the background, which are described by the noise kernel $\xi_{\mu\nu}(x)$. From now on, in \Eq~\eqref{eq:Ein-Lang-lin} we shall focus our attention on the free intrinsic fluctuations $h_{\mu\nu}^{\text{intr},0}$, which satisfy the homogeneous Einstein--Langevin equation, and on the induced fluctuations $h_{\mu\nu}^\text{ind}$, which satisfies the Einstein--Langevin equation with source $\xi_{\mu\nu}(x)$. These are the leading terms in a large-$N$ expansion \cite{Hu2004}.

In the stochastic picture, quantum expectations values really acquire the meaning of statistical averages $\Estoch$. In particular, the stochastic noise is completely characterized by a vanishing expectation value and non-vanishing variance
\begin{equation}
\label{eq:Gaussian-noise}
\Estoch \xi_{\mu\nu}(x) = 0 \eqend{,} \qquad \Estoch \xi_{\mu\nu}(x) \xi_{\rho\sigma}(x') = \mathcal{N}_{\mu\nu\rho\sigma}(x,x') \eqend{.}
\end{equation}
Denoting with $\omega^\text{c}(a b) = \omega(a b) - \omega(a) \omega(b)$ the connected two-point correlation of the state $\omega$, the bitensor $K_{\mu\nu\rho\sigma}(x,x')$ describing the statistical fluctuations associated to the renormalized quantum stress-energy tensor is given by
\begin{equation}
\label{eq:noise-kernel}
K_{\mu\nu\rho\sigma}(x,x') \doteq \omega^\text{c} \left( T_{\mu\nu}(x) T_{\rho\sigma}(x') \right) = \omega \left( t_{\mu\nu}(x) t_{\rho\sigma}(x') \right) \eqend{,}
\end{equation}
where $t_{\mu\nu}(x) \doteq T_{\mu\nu}(x) - \omega\left( T_{\mu\nu}\right) \1$ denotes the centered quantum stress-energy tensor (with vanishing expectation). Thus, \Eq~\eqref{eq:noise-kernel} may be split into its real and imaginary part
\begin{equation}
K_{\mu\nu\rho\sigma}(x,x') = \frac{1}{2} \omega\left( \{ t_{\mu\nu}(x), t_{\rho\sigma}(x') \} \right) + \frac{1}{2} \omega\left( [ t_{\mu\nu}(x), t_{\rho\sigma}(x') ] \right) \eqend{,}
\end{equation}
where $\{a,b\}$ and $[a,b]$ denote the anticommutator (which is real) and the commutator (which is purely imaginary), respectively. Since $\xi_{\mu\nu}$ is a classical stochastic variable, it variance given in \Eq~\eqref{eq:Gaussian-noise} is fully determined by the symmetric part
\begin{equation}
\label{eq:noise-kernel-ind}
\mathcal{N}_{\mu\nu\rho\sigma}(x,x') \doteq \frac{1}{2} \omega^\text{c} \left( \{ T_{\mu\nu}(x), T_{\rho\sigma}(x') \} \right) = \frac{1}{2} K_{\mu\nu\rho\sigma}(x,x') + \frac{1}{2} K_{\rho\sigma\mu\nu}(x',x) \eqend{,}
\end{equation}
thus describing the induced part of the quantum fluctuations associated to the backreaction of the quantum matter field.

Denoting with $\nabla^\mu$ the covariant derivative constructed out of the background metric, the noise $\xi_{\mu\nu}(x)$ is covariantly conserved, i.e., $\nabla^\mu \xi_{\mu\nu}(x) = 0$, since $\nabla^\mu K_{\mu\nu\rho\sigma}(x,x') = 0$ follows from the conservation of the quantum stress-energy tensor, and thus both the equalities $\Estoch \nabla^\mu \xi_{\mu\nu}(x) = 0$ and $\Estoch \nabla^\mu \xi_{\mu\nu}(x) \xi_{\rho\sigma}(x') = 0$ hold. Moreover, as $\xi_{\mu\nu}(x)$ does not depend on $h_{\mu\nu}(x)$, the Einstein--Langevin equation is gauge-invariant whenever the background solution satisfies the semiclassical Einstein equation. For further details, see \cite{Phillips2000,Hu2002,Hu2008,Hu2020stoc}.

\section{Cosmological thermal theory}
\label{sec:photon}

\subsection{The free thermal photon field}

The aim of this section is to describe the quantum matter field of our model, i.e., a free Maxwell field $A_\mu$ in a conformally thermal state, which drives the backreaction of the induced fluctuations in the cosmological background $(\mathcal{M},g^{(0)}_{\mu\nu})$. The quantization of a free Maxwell field has been deeply studied in literature from different perspectives \cite{Furlani1995,Moretti1996,Furlani1999,Pfenning2009,Dappiaggi2011,Finster2013,Hack2013,Glavan2024}. In the framework of algebraic quantum field theory (see, e.g., \cite{Brunetti2015,Hack2016cosm} and references therein), the quantization in globally hyperbolic spacetimes is carried out by constructing the free unital $*$-algebra of observables $\mathfrak{A}_A(\mathcal{M})$ generated by smeared one-form fields $A(f)$ defined for all smooth, compactly supported, complex valued one-form test functions $f \in \Omega_0^1(\mathcal{M})$. To overcome the issue of gauge freedom and define gauge-invariant observables, one defines $A(f)$ as equivalence classes of one-form fields acting on co-closed one-form test functions, i.e., satisfying $\delta f = 0$, where $\delta$ denotes the co-differential operator on $p$-forms. Hence, the quotient $*$-algebra $\mathfrak{A}_A(\mathcal{M})$ generated by $A(f)$ satisfies the following relations, which holds for all co-closed $f,g \in \Omega_0^1(\mathcal{M})$:
\begin{enumerate}
\item Linearity: $A(c_0 f + c_1 g) = c_0 A(f) + c_1 A(g)$, $c_0, c_1 \in \mathbb{C}$.
\item Hermiticty: $A^*(f) = A(\Bar{f})$.
\item On-shell fields: $A(\nabla^2 f) = 0$.
\item Canonical commutations relations: $[A(f), A(f')] = \mathi G_A(f,f') \1$.
\end{enumerate}
Here, $G_A \doteq G_A^\text{adv} - G_A^\text{ret}$ denotes the unique causal (advanced minus retarded) Green operator. 

Denoting with $A_\mu$ the one-form field and specializing to FLRW spacetimes \eqref{eq:FLRW} , the free Maxwell field is conformally invariant in the generalization of the Lorentz gauge $\nabla^\mu A_\mu = 2 H a A_0$~\cite{huguetrenaud2013} (see also \cite{glavan2025}), and the equation of motion simplifies to $\partial^2 A_\mu = 0$. To model thermal matter, we shall consider mixed quasi-free (Gaussian) conformally thermal states $\omega_{\mu\nu,\beta} \colon \mathfrak{A}_A(\mathcal{M}) \to \mathbb{C}$, whose two-point function in $n = 4$ dimensions is given by \cite{Dappiaggi2010,Hack2016cosm,Dappiaggi2017}
\begin{equation}
\label{eq:thermal-state}
\omega_{\mu\nu,\beta}(x,x') = \eta_{\mu\nu} \, \omega_\beta(x,x') \eqend{,}
\end{equation}
where
\begin{equation}
\label{eq:thermal-state-scal-conf}
\omega_\beta(x,x') \doteq \frac{1}{a(\eta) a(\eta')} \int \frac{1}{2 \abs{\vec{p}}} \left[ \frac{\mathe^{- \mathi \abs{\vec{p}} (\eta-\eta')}}{1 - \mathe^{- \beta \abs{\vec{p}}}} + \frac{\mathe^{\mathi \abs{\vec{p}} (\eta-\eta')}}{\mathe^{\beta \abs{\vec{p}}} - 1} \right] \mathe^{\mathi \vec{p} ( \vec{x} - \vec{x}' )} \frac{\total^3 \vec{p}}{(2 \pi)^3}
\end{equation}
is the scalar two-point function of a conformally thermal state. Here, $\beta = (k_\text{B} T)^{-1}$ is the normalized inverse temperature, with $k_\text{B}$ the Boltzmann constant. As expected, in the limit of zero temperature $\beta \to \infty$ the thermal two-point function reduces to the two-point function of the conformal vacuum state $\omega_{\mu\nu,\infty}: \mathfrak{A}_A \to \mathbb{C}$, which is given by
\begin{equation}
\label{eq:conformal-vacuum-omega}
\omega_{\mu\nu, \infty}(x,x') = \eta_{\mu\nu} \, \omega_\infty(x,x')
\end{equation}
with
\begin{equation}
\label{eq:conformal-vacuum-omega-scal}
\omega_{\infty}(x,x') = \frac{1}{a(\eta) a(\eta')} \int \frac{\mathe^{- \mathi \abs{\vec{p}} (\eta-\eta') + \mathi \vec{p} ( \vec{x} - \vec{x}' )}}{2 \abs{\vec{p}}} \frac{\total^3 \vec{p}}{(2 \pi)^3} \eqend{.}
\end{equation}
This class of states is called conformally thermal because they correspond to thermal equilibrium states after conformal rescaling \cite{Dappiaggi2010, Hack2016cosm}. Namely, they satisfy the KMS condition
\begin{equation}
a(\eta-\mathi \beta) a(\eta') \, \omega_\beta(\eta-\mathi \beta,\vec{x},x') = a(\eta) a(\eta') \, \omega_\beta(x',x) \eqend{,}
\end{equation}
as can be easily inferred from the explicit expression~\eqref{eq:thermal-state-scal-conf}. We note that since the background spacetime is time-dependent, it is impossible to obtain a global thermal equilibrium state, but local thermal equilibrium might be possible~\cite{Buchholz2001,Becattini2014,Gransee2015}.

Thanks to the Hadamard property fulfilled by both $\omega_{\mu\nu,\beta}(x,x')$ and $\omega_{\mu\nu,\infty}(x,x')$, their difference is a smooth function, namely
\begin{splitequation}
\label{eq:omega-diff}
\left( \Delta \omega \right)_{\mu\nu}(x,x') &\doteq \omega_{\mu\nu,\beta}(x,x') - \omega_{\mu\nu,\infty}(x,x') \in C^{\infty}(\mathcal{M}) \\
&= \eta_{\mu\nu} \frac{1}{a(\eta) a(\eta')} \int \frac{1}{\abs{\vec{p}}} \frac{\cos\left[ \abs{\vec{p}} (\eta-\eta') \right]}{\mathe^{\beta \abs{\vec{p}}} - 1} \mathe^{\mathi \vec{p} ( \vec{x} - \vec{x}' )} \frac{\total^3 \vec{p}}{(2 \pi)^3} \eqend{.}
\end{splitequation}
In particular, the difference is finite in the coincidence limit $x' \to x$, and we will use the fact to compute the expectation value of the quantum stress-energy tensor.

\subsection{The quantum stress-energy tensor}
\label{sec:stress-tensor}

Using the thermal state given in \Eq~\eqref{eq:thermal-state} constructed out of the conformal vacuum, we may evaluate the expectation value of the renormalized quantum stress-energy tensor $\omega\left( T_{\mu\nu} \right)$, whose classical expression is given by 
\begin{equation}
\label{eq:stress-tensor}
T^\text{cl}_{\mu\nu} = F_{\mu\rho} F_{\nu\sigma} g^{\rho\sigma} - \frac{1}{4} g_{\mu\nu} g^{\alpha\beta} g^{\rho\sigma} F_{\alpha\rho} F_{\beta\sigma} \eqend{,}
\end{equation}
where $F_{\mu\nu} = \nabla_\mu A_\nu - \nabla_\nu A_\mu$ denotes the electromagnetic (Faraday) tensor.

It is well-known that expectation values of quadratic (and other composite) observables generally diverge in the coinciding point limit, a problem which is solved in flat spacetime by defining them via normal ordering. The generalization of this procedure to curved spacetimes proceeds by point splitting and subtracting the divergent term, for which it is necessary that the states that are considered have a universal singular behavior. Therefore, one employs the class of Hadamard states, which are physically reasonable states whose divergent behavior mimics the Minkowski one~\cite{DeWitt1960,Wald1977back,Fewster2013}. In these states, using the point-splitting procedure one obtains a finite expectation value for the quantum stress-energy tensor \cite{Wald1995qftcurved}. Moreover, this procedure works on generic (globally hyperbolic) spacetimes, thus providing a local and covariant definition of Wick observables in curved spacetimes \cite{Brunetti1999,Hollands2001,Hollands2002}. For Maxwell fields (in a suitable gauge), a quasi-free state fulfills the Hadamard condition whenever its two-point function is given by \cite{DeWitt1960,Brunetti1995,Radzikowski1996,Brown1986,Belokogne2015}
\begin{equation}
\omega_{\mu\nu}(x',x) = \mathcal{H}_{\mu\nu}(x',x) + w_{\mu\nu}(x',x) \eqend{,}
\end{equation}
where 
\begin{equation}
\label{eq:hadamard-parametrix}
\mathcal{H}_{\mu\nu}(x,x') = \lim_{\varepsilon \to 0^+} \mathcal{H}_{\mu\nu,\varepsilon}(x,x') = \frac{u_{\mu\nu}(x,x')}{\sigma_\varepsilon(x,x')} + v_{\mu\nu}(x,x') \ln \left( \frac{\sigma_\varepsilon(x,x')}{\mu^2} \right) \eqend{.}
\end{equation}
denotes the Hadamard parametrix (the limit is taken in the usual sense of distributions). Here, $\sigma_\varepsilon(x,x') \doteq \sigma(x,x') + 2 \mathi \epsilon( t(x)-t(x') ) + \varepsilon^2$ with the Synge world function $\sigma(x,x')$ given by one half of the geodesic distance, $t$ is an arbitrary time function, and $u_{\mu\nu}(x,x')$, $v_{\mu\nu}(x,x')$, and $w_{\mu\nu}(x,x')$ denote the Hadamard coefficients, while $\mu$ is a scale introduced to make the logarithm dimensionless. In Hadamard states, renormalized quadratic observables and higher Wick powers may be evaluated by subtracting $\mathcal{H}_{\mu\nu}(x,x')$ before computing the coinciding point limit. Furthermore, in conformally flat spacetimes such as FLRW spacetimes, where a conformal null boundary is attached to the bulk spacetime, the Hadamard property defined on the boundary may be extended to quantum states in the bulk. This ``bulk-to-boundary correspondence'' allows to obtain cosmological Hadamard state which respects the isometries of the spacetime, such as the conformal vacuum and the thermal state given in \Eqs~\eqref{eq:conformal-vacuum-omega}, \eqref{eq:thermal-state}, and consequently to define the algebra of Wick polynomials in this class of spacetimes \cite{Dappiaggi2007,Dappiaggi2008,Dappiaggi2010,Dappiaggi2017}.

While the point-splitting procedure gives a unique result, it might have some undesired features. For example, the quantum stress-energy tensor is not conserved~\cite{Wald1977back}, because the Hadamard parametrix~\eqref{eq:hadamard-parametrix} does not fulfill the equation of motion. However, imposing that the renormalization scheme is local and covariant leaves some renormalization freedom which can be exploited to fulfill further conditions. In fact, the renormalization freedom was fully classified in globally hyperbolic spacetimes \cite{Hollands2001,Hollands2005,Hack2016cosm}, and for the quantum stress-energy tensor \eqref{eq:stress-tensor} is has been proved that two different regularization schemes giving $\tilde{T}_{\mu\nu}$ and $T_{\mu\nu}$ are related by
\begin{equation}
\label{eq:renorm-freed}
\tilde{T}_{\mu\nu} = T_{\mu\nu} + c_0 g_{\mu\nu} + c_1 G_{\mu\nu} + \gamma_1 A_{\mu\nu} + \gamma_2 B_{\mu\nu} \eqend{,}
\end{equation}
where $c_0$, $c_1$, $\gamma_1$, $\gamma_2 \in \mathbb{R}$. These correspond to the freedom of performing finite redefinitions of the bare constants $\Lambda$, $G_N$, $\alpha_A$, $\alpha_B$ appearing in \Eqs~\eqref{eq:Ein-Lang}, respectively, with $A_{\mu\nu}$ and $B_{\mu\nu}$ defined in \Eq~\eqref{eq:amunu_bmunu}. Based on this, different regularization schemes correspond to redefinitions of the constants $\Lambda$, $\alpha_A$ and $\alpha_B$ in the Einstein--Langevin equation~\eqref{eq:Ein-Lang}.

\begin{remark}
As pointed out in \cite{Brown1986} (see also \cite{Belokogne2015}), the definition of the quantum stress-energy tensor in the complete gauge-fixed theory involves two further contributions arising from the gauge-fixing function in the Maxwell action, depending on the gauge-breaking and the compensating ghost term. However, it may be shown that those quantum contributions exactly cancel each other. Therefore the full renormalized quantum stress-energy tensor turns to be exactly equal to the expectation value of $\omega\left( T_{\mu\nu} \right)$ obtained by the classical expression given in \Eq~\eqref{eq:stress-tensor}. In fact, this result holds on any BRST-invariant state, such as the conformal vacuum state \eqref{eq:conformal-vacuum-omega} and its thermal counterpart \eqref{eq:thermal-state}.
\end{remark}

On the other hand, imposing that $\omega\left( T_{\mu\nu} \right)$ is covariantly conserved gives rise to an anomalous contribution in the quantum trace, because the Hadamard parametrix~\eqref{eq:hadamard-parametrix} does not fulfill the equation of motion. Hence, $g^{\mu\nu} T_{\mu\nu}$ does not vanish anymore for conformally invariant fields, but it is (for free fields) a c-number functional of the metric \cite{Wald1978,Moretti2003,Hollands2005}. For Maxwell fields $A_\mu$ in four dimensions it reads \cite{Brown1986,Belokogne2015} 
\begin{equation}
g^{\mu\nu} \omega\left( T_{\mu\nu}(x) \right) = \frac{1}{4 \pi^2} \lim_{x' \to x} \left[ g^{\mu\nu} v_{1\mu\nu}(x,x') - 2 v_1(x,x') \right] \eqend{,}
\end{equation}
where $v_{1\mu\nu}(x,x')$ is one of the vector Hadamard coefficients, and $v_1(x,x')$ one of the scalar Hadamard coefficients. Their coincidence limits are quadratic in curvature tensors, namely
\begin{equation}
\label{eq:trace-anomaly-v}
g^{\mu\nu} \omega\left( T_{\mu\nu} \right) = \frac{1}{8 \pi^2} \left( - \frac{31}{540} R^2 + \frac{31}{180} R_{\mu\nu} R^{\mu\nu} - \frac{13}{360} C_{\mu\nu\rho\sigma} C^{\mu\nu\rho\sigma} - \frac{1}{120} \nabla^2 R \right) \eqend{.}
\end{equation}
However, as the trace anomaly is proportional to the identity, it does not influence the stochastic noise tensor $\xi_{\mu\nu}$ in \Eq~\eqref{eq:Ein-Lang}, because $g^{\mu\nu} K_{\mu\nu\rho\sigma} = 0$, such that $\Estoch g^{\mu\nu} \xi_{\mu\nu}(x) = 0$ and $\Estoch g^{\mu\nu} \xi_{\mu\nu}(x) \xi_{\rho\sigma}(x') = 0$. Therefore, for conformal fields it holds that $g^{\rho\sigma} \xi_{\rho\sigma} = 0$, so no contributions arise in the traced part of \Eqs~\eqref{eq:Ein-Lang} in this case.

\begin{remark}
The contribution proportional to $\nabla^2 R$ in \Eq~\eqref{eq:trace-anomaly-v} is not a true trace anomaly term, because it corresponds to a renormalization freedom of $g^{\mu\nu} \omega_\infty\left( T_{\mu\nu} \right)$ \cite{Hollands2001,Hollands2005}. Thus, in our case it may be always remove by a judicious choice of the renormalization constant $\gamma_2$ in \Eq~\eqref{eq:renorm-freed}, or equivalently $\alpha_B$ in \Eq~\eqref{eq:Ein-Lang}. In fact, the zeroth-order equation \eqref{eq:Ein-Lang-zero} is exactly the semiclassical Einstein equation, and the contribution proportional to $\nabla^2 R$ corresponds to a renormalization of the coefficient $\alpha_B$ related to the trace of $B_{\mu\nu}$ (see \Eq~\eqref{eq:bmunu-trace}). Since in cosmological spacetimes there is a unique degree of freedom given by $a$, the traced semiclassical equation may be promoted to the unique dynamical solution for the system, and in the case of massless conformally fields there is always sufficient freedom in that equation to discard the higher-order derivative term $\nabla^2 R$ \cite{Pinamonti2011,Pinamonti2013}. This is consistent with the perturbative approach to the semiclassical equation based on order reduction, according to which the higher-order derivative terms should be discarded as higher-order perturbative contributions in $\kappa^2$ \cite{Simon1991,Parker1993} (see also \cite{Froeb2013,Glavan2017,Glavan2024red}). We will do this in the following, and choose $\alpha_B$ such that the contribution proportional to $\nabla^2 R$ in \Eq~\eqref{eq:trace-anomaly-v} is canceled.
\end{remark}

In the cosmological case, we may exploit the trace anomaly and the conservation equation $\nabla^\mu \omega\left( T_{\mu\nu} \right) = 0$ to evaluate the renormalized quantum stress-energy tensor in a thermal state, starting from the conformal vacuum state \eqref{eq:conformal-vacuum-omega}. In fact, due to symmetries of FLRW spacetimes and thanks to the conformal invariance of Maxwell fields, the quantum stress-energy tensor is purely geometric in the conformal vacuum.
\begin{proposition}
\label{prop:vacuum}
Consider a FLRW spacetime $(M, g^{(0)})$ \eqref{eq:FLRW} in conformal coordinates $(\eta,\vec{x})$, with scale factor $a \in C^n(I_\eta)$, $n \geq 2$, for some time interval $I_\eta \subseteq \mathbb{R}$, and the conformal vacuum state $\omega_{\mu\nu,\infty}: \mathfrak{A}_A \to \mathbb{C}$ whose two-point function is given in \Eqs~\eqref{eq:conformal-vacuum-omega}, \eqref{eq:conformal-vacuum-omega-scal}. Then
\begin{equation}
\label{eq:stress-tensor-inf}
\omega_\infty\left( T_{\mu\nu} \right) (x) = \frac{31}{1440 \pi^2} \left( - 3 \eta_{\mu\nu} + 4 \epsilon \Bar{\eta}_{\mu\nu} \right) H^4 a^2 \eqend{,}
\end{equation}
where  
\begin{equation}
\label{eq:spat-met}
\Bar{\eta}_{\mu\nu} \doteq \eta_{\mu\nu} + \delta^0_\mu \delta^0_\nu
\end{equation}
denotes the spatial flat metric.
\end{proposition}
\begin{proof}
Since the conformal vacuum state $\omega_{\infty,\mu\nu}$ for free Maxwell fields respects the symmetries of the spacetime, it holds on FLRW spacetimes that
\begin{equations}
\label{eq:stress-tensor-FLRW}
\omega_\infty\left( T_{\mu\nu} \right) = A(\eta) \eta_{\mu\nu} + B(\eta) \delta_\mu^0 \delta_\nu^0 \eqend{,}
\end{equations}
where $A \colon I_\eta \to C^n(I_\eta)$ and $B \colon I_\eta \to C^n(I_\eta)$ are temporal functions constructed out of the scale factor which should be determined. In particular, the traced part of $\omega_\infty\left( T_{\mu\nu} \right)$ is completely fixed by the trace anomaly given in \Eq~\eqref{eq:trace-anomaly-v}. Recalling that the Weyl tensor vanishes in conformally flat spacetimes and using the explicit expressions for the curvature tensors in \Eq~\eqref{eq:curv-FLRW}, we obtain
\begin{equation}
\eta^{\mu\nu} \omega_\infty\left( T_{\mu\nu} \right) = 4 A - B = \frac{31}{4320 \pi^2} a^2 \left( 3 a^{-4} \eta^{\mu\rho} \eta^{\nu\sigma} R_{\mu\nu} R_{\rho\sigma} - R^2 \right) = - \frac{31}{120 \pi^2} (1-\epsilon) H^4 a^2 \eqend{.}
\end{equation}
Moreover, the conservation equation entails
\begin{splitequation}
0 &= a^2 \nabla^\mu \omega_\infty\left( T_{\mu\nu} \right) \\
&= \eta^{\mu\rho} \partial_\rho \omega_\infty\left( T_{\mu\nu} \right) - 2 H a \omega_\infty\left( T_{0\nu} \right) - \delta_\nu^0 H a \eta^{\rho\sigma} \omega_\infty\left( T_{\rho\sigma} \right) \\
&= ( A' - B' - 2 H a A - H a B ) \delta_\nu^0 \eqend{.}
\end{splitequation}
Hence, combining both the equations we get
\begin{equation}
A = - \frac{31}{1440 \pi^2} (3-4\epsilon) H^4 a^2 \eqend{,} \qquad B = \frac{31}{360 \pi^2} \epsilon H^4 a^2 \eqend{,}
\end{equation}
and finally we obtain \Eq~\eqref{eq:stress-tensor-FLRW} recalling the definition of the spatial metric in \Eq~\eqref{eq:spat-met}.
\end{proof}

The expectation value in the thermal state may now be obtained by employing Proposition \ref{prop:vacuum} for the vacuum part, and evaluating the difference
\begin{splitequation}
\label{eq:omega-diff-stress}
\Delta \omega\left( T_{\mu\nu}(x) \right) &\doteq \omega_\beta\left( T_{\mu\nu}(x) \right) - \omega_\infty\left( T_{\mu\nu}(x) \right) \\
&= a^{-2}(\eta) \lim_{x' \to x} \left[ \mathfrak{D}_{\mu\nu}{}^{\gamma\delta}(x,x') \Delta \omega_{\gamma\delta}(x,x') \right] \eqend{,}
\end{splitequation}
where $\Delta \omega_{\gamma\delta}(x,x')$ is defined in \Eq~\eqref{eq:omega-diff} and 
\begin{equation}
\label{eq:gD-op}
\mathfrak{D}_{\mu\nu}{}^{\gamma \delta}(x,x') \doteq \eta^{\gamma \delta} \partial_\mu^x \partial_\nu^{x'} - 2 \delta_{(\mu}^\gamma \partial_{\nu)}^x \partial_\nu^{x'} + \frac{1}{2} \eta_{\mu\nu} \partial_x^\gamma \partial_{x'}^{\delta} + \left( \delta_\mu^\gamma \delta_\nu^\delta - \frac{1}{2} \eta_{\mu\nu} \eta^{\gamma\delta} \right) \partial_x^\rho \partial^{x'}_\rho
\end{equation}
is the bi-differential operator obtained from the classical form of the Maxwell stress-energy tensor \eqref{eq:stress-tensor}.

\begin{proposition}
\label{prop:thermal}
Under the same hypothesis as in Proposition \ref{prop:vacuum} for FLRW spacetimes and taking into account the thermal state $\omega_{\mu\nu,\beta} \colon \mathfrak{A}_A \to \mathbb{C}$ whose two-point function is given in \Eqs~\eqref{eq:thermal-state}, \eqref{eq:thermal-state-scal-conf}, the thermal expectation value of the quantum-stress-energy tensor of the free Maxwell field on cosmological spacetimes reads
\begin{equations}[eq:stress-tensor-thermal]
\omega_\beta\left( T_{00}(x) \right) &= \frac{31}{480 \pi^2} H^4(\eta) a^2(\eta) + \frac{\pi^2}{15 \beta^4} a^{-2}(\eta) \eqend{,} \\
\omega_\beta\left( T_{0i}(x) \right) &= 0 \eqend{,} \\
\omega_\beta\left( T_{ij}(x) \right) &= \left[ - \frac{31}{1440 \pi^2} \left( 3 - 4\epsilon \right) H^4(\eta) a^2(\eta) + \frac{\pi^2}{45 \beta^4} a^{-2}(\eta) \right] \delta_{ij} \eqend{.}
\end{equations}
\end{proposition}
\begin{proof}
The proof consists on evaluating the point-split expression for $\Delta \omega \left( T_{\mu\nu}(x) \right)$ given in \Eq~\eqref{eq:omega-diff-stress}, and adding to it the components of the quantum stress-energy tensor in the conformal vacuum state given in \Eq~\eqref{eq:stress-tensor-FLRW}. Contracting $\mathfrak{D}_{\mu\nu}{}^{\gamma \delta} \eta_{\gamma\delta}$ with the explicit expression for the thermal two-point function given in \Eq~\eqref{eq:thermal-state} yields
\begin{equation}
\label{eq:delta_omega_tmunu}
\Delta \omega\left( T_{\mu\nu}(x) \right) = 2 a^{-2}(\eta) \lim_{x' \to x} \left[ \partial_\mu^x \partial_{\nu}^{x'} \left( \Delta \omega\right)(x',x) - \frac{1}{4} \eta_{\mu\nu} \eta^{\rho\sigma} \partial_\rho^x \partial_{\sigma}^{x'} \left( \Delta \omega\right)(x',x) \right] \eqend{,}
\end{equation}
where
\begin{equation}
\left( \Delta \omega\right)(x,x') \doteq \omega_{\beta}(x,x') - \omega_{\infty}(x,x') = \int \frac{1}{\abs{\vec{p}}} \frac{\cos\left[ \abs{\vec{p}} (\eta-\eta') \right]}{\mathe^{\beta \abs{\vec{p}}} - 1} \mathe^{\mathi \vec{p} ( \vec{x}-\vec{x}' )} \frac{\total^3 \vec{p}}{(2 \pi)^3} \eqend{.}
\end{equation}
Note that
\begin{equation*}
\partial^2_x \left( \Delta \omega\right)(x,x') = \partial^2_{x'} \left( \Delta \omega\right)(x,x') = \eta^{\rho\sigma} \partial_\rho^x \partial_{\sigma}^{x'} \left( \Delta \omega\right)(x,x') = 0 \eqend{,}
\end{equation*}
such that the second term in~\eqref{eq:delta_omega_tmunu} vanishes. Performing the derivatives and taking the limit, we thus obtain
\begin{equations}
\Delta \omega \left( T_{00}(x) \right) &= 2 a^{-2}(\eta) \int \frac{\abs{\vec{p}}}{\mathe^{\beta \abs{\vec{p}}} - 1} \frac{\total^3 \vec{p}}{(2 \pi)^3} = \frac{\pi^2}{15 \beta^4} a^{-2}(\eta) \eqend{,} \\
\Delta \omega \left( T_{0i}(x) \right) &= 0 \eqend{,} \\
\Delta \omega \left( T_{ij}(x) \right) &= \frac{2}{3} \delta_{ij} a^{-2}(\eta) \int \frac{\abs{\vec{p}}}{\mathe^{\beta \abs{\vec{p}}} - 1} \frac{\total^3 \vec{p}}{(2 \pi)^3} = \delta_{ij} \frac{\pi^2}{45 \beta^4} a^{-2}(\eta) \eqend{,}
\end{equations}
where we used that due to rotational invariance of the integral (after taking the limit $x' \to x$) we could replace $\vec{p}_i \to 0$ (for $T_{0i}$) and $\vec{p}_i \vec{p}_j \to \frac{1}{3} \delta_{ij} \vec{p}^2$ (for $T_{ij}$). Finally, summing up all the contributions in \Eq~\eqref{eq:omega-diff-stress}, we obtain \Eq~\eqref{eq:stress-tensor-thermal}.
\end{proof}

The explicit value~\eqref{eq:stress-tensor-thermal} of the thermal quantum stress-energy tensor allows to solve the zeroth-order semiclassical Einstein equation \eqref{eq:Ein-Lang-zero} in FLRW spacetimes. We recall that the curvature tensor $A_{\mu\nu}$ vanishes in FLRW spacetimes, and, furthermore, the curvature tensor $B_{\mu\nu}$ may be removed by a judicious choice of the $\alpha_B$, taking into account the higher-order derivative contribution proportional to $\nabla^2 R$ in the trace anomaly \eqref{eq:trace-anomaly-v}. Hence, using the explicit expressions for the curvature tensors in \Eq~\eqref{eq:curv-FLRW}, the semiclassical Einstein equation \eqref{eq:Ein-Lang-zero} splits into a temporal and a spatial part as follows:
\begin{equations}[eq:Ein-lang-zero-FLRW-thermal]
3 H^2 a^2 - \Lambda a^2 &= \frac{1}{2} \kappa^2 \left[ \frac{31}{480 \pi^2} H^4 a^2 + \frac{\pi^2}{15 \beta^4} a^{-2} \right] \eqend{,} \label{eq:Ein-lang-zero-FLRW-thermal_00} \\
- (3-2\epsilon) H^2 a^2 + \Lambda a^2 &= \frac{1}{2} \kappa^2 \left[ - \frac{31}{1440\pi^2} \left( 3 - 4\epsilon \right) H^4 a^2 + \frac{\pi^2}{45 \beta^4} a^{-2} \right] \eqend{.} \label{eq:Ein-lang-zero-FLRW-thermal_ij}
\end{equations}
By taking a derivative with respect to conformal time $\eta$ of the first equation, we see that the second equation holds if the first one is fulfilled, and so we only need to solve the first equation, which is a nonlinear differential equation of first order for the scale factor $a$.

\begin{corollary}
Consider a FLRW spacetime $\left( \mathcal{M}, g^{(0)}_{\mu\nu} \right)$ and a quantum Maxwell field whose algebra is $\mathfrak{A}_A(\mathcal{M})$, endowed with the quasi-free quasi-thermal state $\omega_{\mu\nu, \beta} \colon \mathfrak{A}_A(\mathcal{M}) \to \mathbb{C}$ whose two-point function is given in \Eqs~\eqref{eq:thermal-state} and~\eqref{eq:thermal-state-scal-conf}. Consider the expectation value of the quantum stress-energy tensor $\omega_\beta\left( T_{\mu\nu}(x) \right)$ evaluated in the quasi-thermal state, which is given in \Eq~\eqref{eq:stress-tensor-thermal}. The zeroth-order cosmological semiclassical Einstein equation
\begin{equation}
G^{(0)}_{\mu\nu}(x) + \Lambda g^{(0)}_{\mu\nu} = \frac{\kappa^2}{2} \omega_\beta\left( T_{\mu\nu}(x) \right)
\end{equation} 
admits a unique FLRW solution with scale factor $a(\eta) \in C^\infty(I_\eta)$ for a sufficiently short time interval $I_\eta$. For vanishing cosmological constant $\Lambda = 0$, it also admits a perturbative solution with $\eta \in (0,\infty)$ of the form
\begin{equation}
a(\eta) \approx a_0 + \frac{\kappa \pi}{\sqrt{90} \, \beta^2} \eta \eqend{,}
\end{equation}
which is valid if $\left[ \kappa / (12 \beta a_0) \right]^4 \ll 1$. This is the solution for a radiation-dominated universe, \Eq~\eqref{eq:power-law} with $\epsilon = 2$, $H_0 = \frac{\kappa \pi}{\sqrt{90} \, \beta^2}$ and $\eta_0 = - a_0/H_0$.
\end{corollary}
\begin{proof}
Based on the discussion stated before, and recalling the results obtained in the previous proposition, the semiclassical Einstein equation reduces to a unique dynamical equation involving the Hubble function $H$ and the scale factor $a$, namely \Eq~\eqref{eq:Ein-lang-zero-FLRW-thermal_00}. Since $H$ is the first derivative of the scale factor $a$, this is a nonlinear first-order differential equation which has a unique solution for short time intervals according to well-known theorems. To exhibit the existence of a perturbative solution, we rewrite \Eq~\eqref{eq:Ein-lang-zero-FLRW-thermal_00} with $\Lambda = 0$ as
\begin{equation}
(H^2 a^4)(\eta) = \bigl[ a'(\eta) \bigr]^2 = \frac{\kappa^2 \pi^2}{15 \beta^4 \left( 6 - \kappa^2 H^2 \frac{31}{480 \pi^2} \right)} \eqend{.}
\end{equation}
In a first approximation we can neglect the term proportional to $\kappa^2 H^2$ in the denominator and obtain the solution
\begin{equation}
a(\eta) \approx a_0 + \frac{\kappa \pi}{\sqrt{90} \, \beta^2} \eta \eqend{.}
\end{equation}
For this solution, we compute
\begin{equation}
H(\eta) \approx \frac{\kappa \pi}{\sqrt{90} \, \beta^2 a(\eta)^2} \leq \frac{\kappa \pi}{\sqrt{90} \, \beta^2 a_0^2} \eqend{,}
\end{equation}
which gives a self-consistent solution (namely, neglecting the term proportional to $\kappa^2 H^2$ is justified a posteriori) if the stated condition holds. Perturbative corrections to this result can then be obtained straightforwardly.
\end{proof}

\section{Thermal noise kernel}
\label{sec:thermal-noise}

In this section, we evaluate the induced thermal cosmological fluctuations encoded in the kernel \eqref{eq:noise-kernel}, whose symmetric part is the stochastic noise kernel \eqref{eq:noise-kernel-ind}, in the thermal state $\omega_{\mu\nu,\beta}$ given in \Eq~\eqref{eq:thermal-state} for free Maxwell fields. About the noise kernel of conformally coupled scalar fields in conformally flat spacetimes, see \cite{Cho2014,Satin2016,Hu2020stoc} and references therein.

Let us denote with $\xi^\beta(x,x')$ the Gaussian thermal stochastic noise associated to that state and entering the linearized Einstein--Langevin equation \eqref{eq:Ein-Lang-lin} in the thermal case. Let us define the connected quantum stress-energy tensor thermal correlator $t^\beta_{\mu\nu} \doteq T_{\mu\nu} - \omega_\beta \left( T_{\mu\nu} \right) \1$, then $\xi^\beta(x,x')$ is completely characterized by its vanishing expectation value and its non-vanishing induced variance:
\begin{equations}
\Estoch \xi^\beta_{\mu \nu}(x) &= 0 \eqend{,} \\
\Estoch \xi^\beta_{\mu \nu}(x) \xi^\beta_{\rho\sigma} (x') &= \frac{1}{2} \omega_\beta\left( \{ t^\beta_{\mu\nu}(x), t^\beta_{\rho\sigma}(x') \}\right) = \frac{1}{2} K^\beta_{\mu\nu\rho\sigma}(x,x') + \frac{1}{2} K^\beta_{\rho\sigma\mu\nu}(x',x) \eqend{,}
\end{equations}
where we define the thermal variance kernel
\begin{equation}
\label{eq:variance-thermal}
K^\beta_{\mu\nu\rho\sigma}(x,x') \doteq \omega_\beta\left( t^\beta_{\mu\nu}(x) t^\beta_{\rho\sigma}(x') \right) \eqend{.}
\end{equation}

The evaluation of the thermal variance \eqref{eq:variance-thermal} is provided as follows:
\begin{proposition}
\label{prop:variance}
Consider a FLRW spacetime $(\mathcal{M},g^{(0)}_{\mu\nu})$ and the quasi-thermal state $\omega_{\mu\nu,\beta}$ given in \Eq~\eqref{eq:thermal-state} of a free Maxwell field. Define the following two-point correlation function
\begin{splitequation}
\label{eq:thermal-two-point-kernel}
\mathscr{K}_{\rho\sigma\gamma\delta}(x,x') &\doteq \partial^x_\rho \partial^x_\sigma \omega_\beta(x,x') \partial^x_\gamma \partial^x_\delta \omega_\beta(x,x') \\
&= \int_{\mathbb{R}^3} \tilde{K}_{\rho\sigma\gamma\delta}\left( \eta - \eta', \vec{p} \right) \mathe^{\mathi \vec{p} ( \vec{x} - \vec{x}') )} \frac{\total^3 \vec{p}}{(2 \pi)^3} \eqend{,}
\end{splitequation}
where
\begin{splitequation}
\label{eq:kernel-thermal}
&\tilde{K}_{\rho\sigma\gamma\delta}(\eta-\eta',\vec{p}) \\
&\doteq \lim_{\epsilon \to 0^+} \int \frac{1}{4 \abs{\vec{p}-\vec{q}} \abs{\vec{q}}} \left[ \frac{\mathe^{\mathi \abs{\vec{p}-\vec{q}} (\eta-\eta')}}{1 - \mathe^{- \beta \abs{\vec{p}-\vec{q}}}} (p-q)_\rho (p-q)_\sigma + \frac{\mathe^{- \mathi \abs{\vec{p}-\vec{q}} (\eta-\eta')}}{\mathe^{\beta \abs{\vec{p}-\vec{q}}} - 1} (p-q)^*_\rho (p-q)^*_\sigma \right] \\
&\qquad\qquad\times \left[ \frac{\mathe^{\mathi \abs{\vec{q}} (\eta-\eta')}}{1 - \mathe^{- \beta \abs{\vec{q}}}} q_\gamma q_\delta + \frac{\mathe^{- \mathi \abs{\vec{q}} (\eta-\eta')}}{\mathe^{\beta \abs{\vec{q}}} - 1} q^*_\gamma q^*_\delta \right] \mathe^{- \epsilon (\abs{\vec{p}-\vec{q}}+\abs{\vec{q}})} \frac{\total^3 \vec{q}}{(2 \pi)^3}
\end{splitequation}
is defined in the sense of distributions, with $k_\mu \doteq \left( \abs{\vec{k}}, \vec{k}\right)$, $k^*_\mu \doteq \left( -\abs{\vec{k}}, \vec{k} \right)$. Then the induced variance \eqref{eq:variance-thermal} is given by
\begin{equation}
\label{eq:thermal-variance}
K^\beta_{\mu\nu\alpha\beta}(x,x') = 2 \mathcal{E}_{\mu\nu\alpha\beta}{}^{\rho\sigma\gamma\delta}(\eta,\eta') \mathscr{K}_{\rho\sigma\gamma\delta}(x,x') \eqend{,}
\end{equation}
where
\begin{splitequation}
&\mathcal{E}_{\mu\nu\alpha\beta}{}^{\rho\sigma\gamma\delta}(\eta,\eta') \doteq a^{-2}(\eta) a^{-2}(\eta') \biggl[ 2 \delta_{(\mu}^\rho \delta_{\nu)}^\sigma \delta_{(\alpha}^\gamma \delta_{\beta)}^\delta \\
&\hspace{4em}- \eta^{\rho\gamma} \left( 4 \delta_{(\mu}^\sigma \eta_{\nu)(\alpha} \delta_{\beta)}^\delta - \eta_{\mu(\alpha} \eta_{\beta)\nu} \eta^{\sigma\delta}- \eta_{\mu\nu} \delta_{(\alpha}^\sigma \delta_{\beta)}^\delta - \eta_{\alpha\beta} \delta_{(\mu}^\sigma \delta_{\nu)}^\delta + \frac{1}{2} \eta_{\mu\nu} \eta_{\alpha\beta} \eta^{\sigma\delta} \right) \biggr] \eqend{.}
\raisetag{2.6em}
\end{splitequation}
In particular, the kernel \eqref{eq:kernel-thermal} is traceless on each pair of indices,
\begin{equation}
\label{eq:traceless}
\eta^{\rho\sigma} \tilde{K}_{\rho\sigma\gamma\delta}(\eta-\eta',\vec{p}) = \eta^{\gamma\delta} \tilde{K}_{\rho\sigma\gamma\delta}(\eta-\eta',\vec{p}) = 0 \eqend{,}
\end{equation}
and it is symmetric,
\begin{equation}
\label{eq:symmetric}
\tilde{K}_{\rho\sigma\gamma\delta}(\eta-\eta',\vec{p}) = \tilde{K}_{(\rho\sigma)(\gamma\delta)}(\eta-\eta',\vec{p}) = \tilde{K}_{\gamma\delta\rho\sigma}(\eta-\eta',\vec{p}) \eqend{.}
\end{equation}
Furthermore, the kernel of the induced variance is traceless on each pair of indices,
\begin{equation}
\label{eq:induced_traceless}
\eta^{\rho\sigma} K^\beta_{\rho\sigma\gamma\delta}(x,x') = \eta^{\gamma\delta} K^\beta_{\rho\sigma\gamma\delta}(x,x') = 0 \eqend{,}
\end{equation}
it is symmetric,
\begin{equation}
\label{eq:induced_symmetric}
K^\beta_{\rho\sigma\gamma\delta}(x,x') = K^\beta_{(\rho\sigma)(\gamma\delta)}(x,x') = K^\beta_{\gamma\delta\rho\sigma}(x,x') \eqend{.}
\end{equation}
and it is conserved,
\begin{equation}
\label{eq:induced_conserved}
\nabla^\rho K^\beta_{\rho\sigma\gamma\delta}(x,x') = \partial^\rho \left[ a^2(\eta) a^2(\eta') K^\beta_{\rho\sigma\gamma\delta}(x,x') \right] = 0 \eqend{.}
\end{equation}
\end{proposition}

\begin{proof}
Recalling the definition of the two-point thermal correlator in Section \ref{sec:stochastic}, the proof is accomplished by evaluating
\begin{equation}
\omega_\beta\left( t^\beta_{\mu\nu}(x) t^\beta_{\rho\sigma}(x') \right) = \mathcal{T}_{\mu\nu}{}^{\rho\sigma\zeta\tau}(x) \mathcal{T}_{\alpha\beta}{}^{\gamma\delta\xi\upsilon}(x') \, \omega_\beta^\text{c}\left( F_{\zeta\rho}(x) F_{\tau\sigma}(x) F_{\xi\gamma}(x') F_{\upsilon\delta}(x') \right) \eqend{,}
\end{equation}
where
\begin{equation}
\mathcal{T}_{\mu\nu}{}^{\rho\sigma\zeta\tau}(x) = a^{-2}(x) \left( \delta_\mu^{(\rho} \delta_\nu^{\sigma)} - \frac{1}{4} \eta_{\mu\nu} \eta^{\rho\sigma} \right) \eta^{\zeta \tau}
\end{equation}
is the operator obtained from the definition of the classical stress-energy tensor \eqref{eq:stress-tensor}. Recalling \Eqs~\eqref{eq:thermal-state}, \eqref{eq:thermal-state-scal-conf} and using that
\begin{equation}
\omega_\beta\left( F_{\tau\sigma}(x) F_{\upsilon\delta}(y) \right) = 4 \delta_{[\tau}^\alpha \delta_{\sigma]}^\beta \delta_{[\upsilon}^\mu \delta_{\delta]}^\nu \partial_\alpha^x \partial_\mu^y \, \omega_\beta\left( A_\beta(x) A_\nu(x') \right) = 4 \partial_{[\tau} \eta_{\sigma][\upsilon} \partial_{\delta]} \omega_\beta(x,x')
\end{equation}
and
\begin{equation}
\label{eq:omegabeta_eom}
\partial^2_x \omega_\beta(x,x') = \partial^2_{x'} \omega_\beta(x,x') = \eta^{\rho\sigma} \partial_\rho^x \partial_\sigma^{x'} \omega_\beta(x,x') = 0 \eqend{,}
\end{equation}
we obtain after a straightforward computation
\begin{splitequation}
\omega_\beta\left( t^\beta_{\mu\nu}(x) t^\beta_{\rho\sigma}(x') \right) &= 2 \mathcal{E}_{\mu\nu\alpha\beta}{}^{\rho\sigma\gamma\delta}(x,x') \left[ \partial^x_\rho \partial^x_\sigma \omega_\beta(x,x') \partial^x_\gamma \partial^x_\delta \omega_\beta(x,x') \right] \eqend{.}
\end{splitequation}
Therefore, \Eq~\eqref{eq:thermal-variance} holds. The explicit expression of the noise kernel in Fourier space given in \Eq~\eqref{eq:kernel-thermal} is straightforwardly obtained by evaluating the spatial Fourier transform of \Eq~\eqref{eq:thermal-two-point-kernel}. The symmetry of this kernel may be easily inferred by shifting $\vec{q}$ with $\vec{p} - \vec{q}$ in \Eq~\eqref{eq:kernel-thermal}, whereas the traceless property may be seen to hold at the integrand level using that $\eta^{\mu\nu} q_\mu q_\nu = \eta^{\mu\nu} q^*_\mu q^*_\nu = 0$ (this actually corresponds to the equation of motion $\partial^2 A_\mu = 0$ for massless photons). Finally, the symmetry, tracelessness and conservation of the kernel of the induced variance is a straightforward computation employing again \Eq~\eqref{eq:omegabeta_eom}.
\end{proof}

Explicit analytical expressions of the thermal noise kernel given in \Eqs~\eqref{eq:thermal-variance}, \eqref{eq:kernel-thermal} are generally difficult to achieve. However, we may derive analytical expressions in certain limits, namely $\beta \to 0^+$ (large temperature), $\beta \to \infty$ (small temperature), and $\abs{\vec{p}} \to 0$ (small external momentum). We shall perform such evaluations in the following subsections. The small external momentum limit on the Maxwell thermal correlator entering the noise kernel has been already studied in \cite{Froeb2025} in the evaluation of the power spectra of the linearized tensorial perturbations, when the superhorizon limit $\abs{\vec{p}} \ll H$ was considered.

\subsection{Small external momentum limit}

An explicit expression of \Eq~\eqref{eq:kernel-thermal} may be provided in the small external momentum limit $\abs{\vec{p}} \to 0^+$. This is the easiest limit to compute, since we can simply expand the integrand for small $\vec{p}$ and compute the resulting integrals. Let us define
\begin{splitequation}
\label{eq:thermal-kernel-beta}
\tilde{K}_\beta(\tau) &\doteq \partial_\tau \ln \sinh\left( \frac{2 \pi \abs{\tau}}{\beta} \right) - \mathi \pi \delta(\tau) + \frac{4 \pi \mathi}{\beta^2} \tau \coth\left( \frac{2 \pi \tau}{\beta} \right) \\
&= \frac{2 \pi}{\beta^2} \lim_{\epsilon \to 0^+} \left[ \bigl[ \beta + 2 \mathi (\tau + \mathi \epsilon) \bigr] \coth\left( \frac{2 \pi (\tau + \mathi \epsilon)}{\beta} \right) \right] \eqend{,}
\end{splitequation}
viewed in the sense of distributions. The explicit computation is performed in App.~\ref{app:smallp}, and the result reads
\begin{equations}[eq:kernel-thermal-smallp]
\tilde{K}^\beta_{0000}(\eta-\eta',\vec{p}) &= \frac{\mathi}{256 \pi^2} \partial_\eta^4 \tilde{K}_\beta(\eta-\eta') + \frac{\pi^2}{15 \beta^5} + \bigo{\vec{p}^2} \eqend{,} \\
\tilde{K}^\beta_{000i}(\eta-\eta',\vec{p}) &= - \frac{\vec{p}_i}{384 \pi^2} \left[ \partial_\eta^3 \tilde{K}_\beta(\eta-\eta') - \frac{64 \mathi \pi^4}{15 \beta^5} \bigl[ 2 (\eta-\eta') - \mathi \beta \bigr] \right] + \bigo{\abs{\vec{p}}^3} \eqend{,} \\
\tilde{K}^\beta_{0i0j}(\eta-\eta',\vec{p}) &= - \frac{1}{3} \eta_{ij} \left[ \tilde{K}_{0000}(\eta-\eta',\vec{p}) - \frac{2 \pi^2}{15 \beta^5} \right] + \bigo{\vec{p}^2} \eqend{,} \\
\tilde{K}^\beta_{00ij}(\eta-\eta',\vec{p}) &= \frac{1}{3} \eta_{ij} \tilde{K}_{0000}(\eta-\eta',\vec{p}) + \bigo{\vec{p}^2} \eqend{,} \\
\begin{split}
\tilde{K}^\beta_{0ikl}(\eta-\eta',\vec{p}) &= \frac{3}{5} \eta_{kl} \tilde{K}_{000i}(\eta-\eta',\vec{p}) - \frac{4}{5} \eta_{i(k} \tilde{K}_{l)000}(\eta-\eta',\vec{p}) \\
&\quad+ \frac{\mathi \pi^2 \left[ 3 \eta_{i(k} \vec{p}_{l)} - \eta_{kl} \vec{p}_i \right]}{225 \beta^5} \bigl[ 2 (\eta-\eta') - \mathi \beta \bigr] + \bigo{\abs{\vec{p}}^3}
\end{split} \\
\tilde{K}^\beta_{ijkl}(\eta-\eta',\vec{p}) &= \frac{\eta_{ij} \eta_{kl} + 2 \eta_{i(k} \eta_{l)j}}{15} \tilde{K}_{0000}(\eta-\eta',\vec{p}) + \bigo{\vec{p}^2} \eqend{.}
\end{equations}

We may also obtain an explicit expression for the noise kernel in coordinate space by using the Fourier transforms
\begin{equation}
\int \mathe^{\mathi \vec{p} \vec{x}} \frac{\total^3 \vec{p}}{(2 \pi)^3} = \delta(\vec{x}) \eqend{,} \quad \int \vec{p}_j \mathe^{\mathi \vec{p} \vec{x}} \frac{\total^3 \vec{p}}{(2 \pi)^3} = - \mathi \partial_j \delta(\vec{x}) \eqend{,}
\end{equation}
which hold in the sense of distributions. It is clear that the limit of small external momentum corresponds to a local approximation of the noise kernel in coordinate space, which is the one that was also employed in \cite{Froeb2025}.

\subsection{Large temperature limit}

The large temperature limit $\beta \to 0^+$ is more difficult, since the integrals in \Eq~\eqref{eq:kernel-thermal} do not converge uniformly in this limit. One has to first rescale the integration variable $\vec{q} \to \vec{q}/\beta$ before performing the expansion for small $\beta$. Let us define the function
\begin{equation}
\tilde{S}(\tau,\vec{p}) \doteq \frac{\sin(\abs{\vec{p}} \tau)}{\abs{\vec{p}} \tau} \eqend{,}
\end{equation}
whose derivatives are given by
\begin{equations}
\partial_\tau \tilde{S}(\tau,\vec{p}) &= \frac{\abs{\vec{p}} \tau \cos\left( \abs{\vec{p}} \tau \right) - \sin\left( \abs{\vec{p}} \tau \right)}{\abs{\vec{p}} \tau^2} \eqend{,} \\
\partial_\tau^2 \tilde{S}(\tau,\vec{p}) &= \frac{- 2 \abs{\vec{p}} \tau \cos\left( \abs{\vec{p}} \tau \right) + \left( 2 - \vec{p}^2 \tau^2 \right) \sin\left( \abs{\vec{p}} \tau \right)}{\abs{\vec{p}} \tau^3} \eqend{,} \\
\partial_\tau^3 \tilde{S}(\tau,\vec{p}) &= \frac{\abs{\vec{p}} \tau \left( 6 - \vec{p}^2 \tau^2 \right) \cos\left( \abs{\vec{p}} \tau \right) - 3 \left( 2 - \vec{p}^2 \tau^2 \right) \sin\left( \abs{\vec{p}} \tau \right)}{\abs{\vec{p}} \tau^4} \eqend{.}
\end{equations}
The explicit computation is performed in App.~\ref{app:larget}, and the result reads
\begin{equations}[eq:kernel-thermal-larget]
\tilde{K}^\beta_{0000}(\eta-\eta',\vec{p}) &= \frac{\pi^2}{15 \beta^5} \tilde{S}(\eta-\eta',\vec{p}) - \frac{\mathi \pi^2}{30 \beta^4} \partial_\eta \tilde{S}(\eta-\eta',\vec{p}) + \bigo{\beta^{-3}} \eqend{,} \\
\tilde{K}^\beta_{000i}(\eta-\eta',\vec{p}) &= - \mathi \frac{\vec{p}_i}{\vec{p}^2} \partial_\eta \tilde{K}_{0000}(\eta-\eta',\vec{p}) + \bigo{\beta^{-3}} \eqend{,} \\
\begin{split}
\tilde{K}^\beta_{0i0j}(\eta-\eta',\vec{p}) &= \frac{\vec{p}_i \vec{p}_j}{\vec{p}^2} \tilde{K}_{0000}(\eta-\eta',\vec{p}) \\
&\quad+ \frac{\eta_{ij} \vec{p}^2 - 3 \vec{p}_i \vec{p}_j}{2 \abs{\vec{p}}^4} \left( \partial_\eta^2 + \vec{p}^2 \right) \tilde{K}_{0000}(\eta-\eta',\vec{p}) + \bigo{\beta^{-3}} \eqend{,}
\end{split} \\
\tilde{K}^\beta_{00ij}(\eta-\eta',\vec{p}) &= \tilde{K}_{0i0j}(\eta-\eta',\vec{p}) + \bigo{\beta^{-3}} \eqend{,} \\
\begin{split}
\tilde{K}^\beta_{0ijk}(\eta-\eta',\vec{p}) &= - \frac{\mathi \vec{p}_i \vec{p}_j \vec{p}_k}{\abs{\vec{p}}^4} \partial_\eta \tilde{K}_{0000}(\eta-\eta',\vec{p}) \\
&\quad- \frac{\mathi \left( 3 \eta_{(ij} \vec{p}_{k)} \vec{p}^2 - 5 \vec{p}_i \vec{p}_j \vec{p}_k \right)}{2 \abs{\vec{p}}^6} \partial_\eta \left( \partial_\eta^2 + \vec{p}^2 \right) \tilde{K}_{0000}(\eta-\eta',\vec{p}) + \bigo{\beta^{-3}} \eqend{,} \raisetag{4.6em}
\end{split} \\
\begin{split}
\tilde{K}^\beta_{ijkl}(\eta-\eta',\vec{p}) &= \frac{\vec{p}_i \vec{p}_j \vec{p}_k \vec{p}_l}{\abs{\vec{p}}^4} \tilde{K}_{0000}(\eta-\eta',\vec{p}) \\
&\quad+ \left[ \frac{3 \eta_{(ij} \vec{p}_k \vec{p}_{l)}}{\abs{\vec{p}}^4} - \frac{5 \vec{p}_i \vec{p}_j \vec{p}_k \vec{p}_l}{\abs{\vec{p}}^6} \right] \left( \partial_\eta^2 + \vec{p}^2 \right) \tilde{K}_{0000}(\eta-\eta',\vec{p}) \\
&\quad+ \left[ \frac{3 \eta_{(ij} \eta_{kl)}}{8 \abs{\vec{p}}^4} - \frac{15 \eta_{(ij} \vec{p}_k \vec{p}_{l)}}{4 \abs{\vec{p}}^6} + \frac{35 \vec{p}_i \vec{p}_j \vec{p}_k \vec{p}_l}{8 \abs{\vec{p}}^8} \right] \left( \partial_\eta^2 + \vec{p}^2 \right)^2 \tilde{K}_{0000}(\eta-\eta',\vec{p}) \\
&\quad+ \bigo{\beta^{-3}} \eqend{.} \raisetag{8em}
\end{split}
\end{equations}

\subsection{Small temperature limit}

The limit of small temperature $\beta \to \infty$ is again relatively easy to compute, and we only have to separate the zero-temperature contributions from the temperature-dependent ones. The explicit computation is performed in App.~\ref{app:smallt}, and the result reads
\begin{equations}[eq:kernel-thermal-smallt]
\begin{split}
\tilde{K}^\beta_{0000}(\tau,\vec{p}) &= \frac{\mathi \mathe^{\mathi \abs{\vec{p}} \tau} \left( 45 - 45 \mathi \abs{\vec{p}} \tau - 20 \vec{p}^2 \tau^2 + 5 \mathi \abs{\vec{p}}^3 \tau^3 + \abs{\vec{p}}^4 \tau^4 \right)}{480 \pi^2} \lim_{\epsilon \to 0^+} \frac{1}{\left( \tau + \mathi \epsilon \right)^5} \\
&\quad+ \mathe^{\mathi \abs{\vec{p}} \tau} \abs{\vec{p}} \frac{\pi^2}{30 \beta^4} + \bigo{\beta^{-6}} \eqend{,}
\end{split} \\
\begin{split}
\tilde{K}^\beta_{000i}(\tau,\vec{p}) &= \vec{p}_i \frac{\mathe^{\mathi \abs{\vec{p}} \tau} \left( 15 - 15 \mathi \abs{\vec{p}} \tau - 7 \vec{p}^2 \tau^2 + 2 \mathi \abs{\vec{p}}^3 \tau^3 \right)}{960 \pi^2} \lim_{\epsilon \to 0^+} \frac{1}{\left( \tau + \mathi \epsilon \right)^4} \\
&\quad+ \vec{p}_i \, \mathe^{\mathi \abs{\vec{p}} \tau}\frac{\pi^2}{60 \beta^4} + \bigo{\beta^{-6}} \eqend{,}
\end{split} \\
\begin{split}
\tilde{K}^\beta_{00ij}(\tau,\vec{p}) &= \eta_{ij} \frac{\mathi \mathe^{\mathi \abs{\vec{p}} \tau} \left( 30 - 30 \mathi \abs{\vec{p}} \tau - 13 \vec{p}^2 \tau^2 + 3 \mathi \abs{\vec{p}}^3 \tau^3 \right)}{960 \pi^2} \lim_{\epsilon \to 0^+} \frac{1}{\left( \tau + \mathi \epsilon \right)^5} \\
&\quad- \vec{p}_i \vec{p}_j \frac{\mathi \mathe^{\mathi \abs{\vec{p}} \tau} \left( 1 - \mathi \abs{\vec{p}} \tau - 2 \vec{p}^2 \tau^2 \right)}{960 \pi^2} \lim_{\epsilon \to 0^+} \frac{1}{\left( \tau + \mathi \epsilon \right)^3} \\
&\quad+ \left( \eta_{ij} + 3 \frac{\vec{p}_i \vec{p}_j}{\vec{p}^2} \right) \mathe^{\mathi \abs{\vec{p}} \tau} \abs{\vec{p}} \frac{\pi^2}{180 \beta^4} + \bigo{\beta^{-6}} \eqend{,}
\end{split} \\
\begin{split}
\tilde{K}^\beta_{0i0j}(\tau,\vec{p}) &= - \eta_{ij} \frac{\mathi \mathe^{\mathi \abs{\vec{p}} \tau} \left( 15 - 15 \mathi \abs{\vec{p}} \tau - 6 \vec{p}^2 \tau^2 + \mathi \abs{\vec{p}}^3 \tau^3 \right)}{480 \pi^2} \lim_{\epsilon \to 0^+} \frac{1}{\left( \tau + \mathi \epsilon \right)^5} \\
&\quad- \vec{p}_i \vec{p}_j \frac{\mathi \mathe^{\mathi \abs{\vec{p}} \tau} \left( 3 - 3 \mathi \abs{\vec{p}} \tau - \vec{p}^2 \tau^2 \right)}{480 \pi^2} \lim_{\epsilon \to 0^+} \frac{1}{\left( \tau + \mathi \epsilon \right)^3} + \bigo{\beta^{-6}} \eqend{,}
\end{split} \\
\begin{split}
\tilde{K}^\beta_{0ikl}(\tau,\vec{p}) &= \left( 3 \vec{p}_i \eta_{kl} - 4 \vec{p}_{(k} \eta_{l)i} \right) \frac{\mathe^{\mathi \abs{\vec{p}} \tau} \left( 3 - 3 \mathi \abs{\vec{p}} \tau - \vec{p}^2 \tau^2 \right)}{960 \pi^2} \lim_{\epsilon \to 0^+} \frac{1}{\left( \tau + \mathi \epsilon \right)^4} \\
&\quad- \vec{p}_i \vec{p}_k \vec{p}_l \frac{\mathe^{\mathi \abs{\vec{p}} \tau} \left( 1 - \mathi \abs{\vec{p}} \tau \right)}{480 \pi^2} \lim_{\epsilon \to 0^+} \frac{1}{\left( \tau + \mathi \epsilon \right)^2} + \vec{p}_i \eta_{kl} \, \mathe^{\mathi \abs{\vec{p}} \tau} \frac{\pi^2}{180 \beta^4} + \bigo{\beta^{-6}} \eqend{,} \raisetag{5.2em}
\end{split} \\
\begin{split}
\tilde{K}^\beta_{ijkl}(\tau,\vec{p}) &= \vec{p}_i \vec{p}_j \vec{p}_k \vec{p}_l \frac{\mathi \mathe^{\mathi \abs{\vec{p}} \tau}}{480 \pi^2} \lim_{\epsilon \to 0^+} \frac{1}{\tau + \mathi \epsilon} \\
&\quad+ \left( 8 \vec{p}_{(i} \eta_{j)(k} \vec{p}_{l)} - 3 \vec{p}_i \vec{p}_j \eta_{kl} - 3 \vec{p}_k \vec{p}_l \eta_{ij} \right) \frac{\mathi \mathe^{\mathi \abs{\vec{p}} \tau} \left( 1 - \mathi \abs{\vec{p}} \tau \right)}{960 \pi^2} \lim_{\epsilon \to 0^+} \frac{1}{\left( \tau + \mathi \epsilon \right)^3} \\
&\quad+ \left( \eta_{ij} \eta_{kl} + 2 \eta_{i(k} \eta_{l)j} \right) \frac{\mathi \mathe^{\mathi \abs{\vec{p}} \tau} \left( 3 - 3 \mathi \abs{\vec{p}} \tau - \vec{p}^2 \tau^2 \right)}{480 \pi^2} \lim_{\epsilon \to 0^+} \frac{1}{\left( \tau + \mathi \epsilon \right)^5} \\
&\quad+ \left( \vec{p}_i \vec{p}_j \eta_{kl} + \vec{p}_k \vec{p}_l \eta_{ij} \right) \mathe^{\mathi \abs{\vec{p}} \tau} \frac{\pi^2}{180 \abs{\vec{p}} \beta^4} + \bigo{\beta^{-6}} \eqend{.} \raisetag{2em}
\end{split}
\end{equations}

It is worth mentioning that the purely vacuum noise kernel, namely the zero-temperature limit $\beta \to \infty$ of \Eq~\eqref{eq:kernel-thermal-smallt}, is analytic in $\vec{p}$. Therefore, it must agree with the zero-temperature limit $\beta \to \infty$ of the small-momentum expansion~\eqref{eq:kernel-thermal-smallp}. Taking the limit $\beta \to \infty$ of the distribution $\tilde{K}_\beta(\tau)$~\eqref{eq:thermal-kernel-beta}, we obtain
\begin{equation}
\lim_{\beta \to \infty} \tilde{K}_\beta(\tau) \doteq \tilde{K}_\infty(\tau) = \partial_\tau \ln\abs{\tau} - \mathi \pi \delta(\tau) = \lim_{\epsilon \to 0^+} \frac{1}{\tau + \mathi \epsilon} \eqend{,}
\end{equation}
and it is an easy exercise to verify that with this kernel the small-momentum expansion~\eqref{eq:kernel-thermal-smallp} agrees with the zero-temperature limit $\beta \to \infty$ of~\eqref{eq:kernel-thermal-smallt} up to order $\vec{p}^2$. Moreover, the difference between the zero-temperature kernel $\tilde{K}_\infty(\tau)$ and the temperature-dependent one $\tilde{K}_\beta(\tau)$ is a smooth function, namely
\begin{equation}
\Delta K(\tau) \doteq \tilde{K}_\beta(\tau) - \tilde{K}_\infty(\tau) = \frac{2 \pi (\beta + 2 \mathi \tau)}{\beta^2} \left[ \coth\left( \frac{2 \pi \tau}{\beta} \right) - \frac{\beta}{2 \pi \tau} \right] + \frac{2 \mathi}{\beta} \in C^\infty(I_\tau) \eqend{.}
\end{equation}
The smoothness of $\Delta K(\tau)$ is in accordance with the analysis given for the quantum stress-energy tensor in Section \ref{sec:stress-tensor} (see \Eqs~\eqref{eq:omega-diff}, \eqref{eq:omega-diff-stress}) as consequence of the property of Hadamard states, which always differ by a smooth function.

\section{Thermal variance and stochastic kernel}

\subsection{The local measurement framework of Fewster and Verch}
\label{sec:measurement}

In this section we aim to analyze the induced thermal fluctuations obtained in the previous section in the algebraic framework of the local and covariant quantum measurement framework proposed by Fewster and Verch in curved spacetimes \cite{Fewster2018,Fewster2025}. Preliminarily, we shall review some main results of that paper which will be essential for our future analysis.

In this framework, one assumes that there exist two linear quantum fields $\Phi$, $\Psi$ generating two abstract $*$-algebras $\mathfrak{A}(\mathcal{M})$, $\mathcal{B}(\mathcal{M})$ (of arbitrary nature), describing respectively the target system and the probe. Measurements are performed on the target system through the probe by means of an interaction, such as a coupling between $\Phi$ and $\Psi$, which is localized in a bounded region of the spacetime $K \subset \mathcal{M}$. In the perturbative algebraic framework the target-probe system is mathematically represented by a full action $S \doteq S_0 + S_I$, where $S_0$ denotes the free action of the quantum fields, and 
\begin{equation}
\label{eq:action-int}
S_\text{int} \doteq - \lambda \int_{\mathcal{M}} \mathfrak{L}_\text{int}(x) \varrho(x) \total\mu_x, \qquad \lambda \in \mathbb{R} \eqend{,}
\end{equation}
is the interaction constructed out of the interaction Lagrangian $\mathfrak{L}_I$. The compactly supported function $\varrho(x) \in C^{\infty}_0(\mathcal{M})$ is a cutoff function supported on the spacetime region $K$ where the interaction takes place, which might be eventually removed in the adiabatic limit $\varrho \to 1$.

As the measurement is bounded in a compact spacetime region, one may identify in $(-)$ and out $(+)$ regions by $\mathcal{M}^{\pm} \doteq \mathcal{M} \, \backslash \, J^\mp(K)$, in which the target and probe fields are uncoupled, and thus described by the free action $S_0$. In particular, one assumes that in the out region a measurement of the probe $\Psi$ is performed in order to obtain information on the target $\Phi$. In those regions where the target and the probe fields are uncoupled, the full quantum system is modeled by the tensor product algebra $\mathfrak{A}(\mathcal{M}) \otimes \mathfrak{B}(\mathcal{M})$, whose generators are of the form
\begin{equation}
\Xi \doteq \Phi \otimes \1_{\mathfrak{A}(\mathcal{M})} + \1_{\mathfrak{B}(\mathcal{M})} \otimes \Psi \eqend{.}
\end{equation}
If $\Phi$ and $\Psi$ satisfy some Green-hyperbolic equations of motion $P \Phi = 0$, $Q \Psi = 0$, then within $\mathcal{M}^\pm$ the full quantum system satisfies the equation of motion $T \Xi = 0$, where $T \doteq P \oplus Q$ is also Green-hyperbolic. Thus, it admits unique advanced and retarded fundamental solutions just as $P$ and $Q$ do.

On the contrary, in the region $K$ of the experiment, the quantum theory is described by an interacting $*$-algebra $\mathfrak{C}(\mathcal{M})$ in which the target and the probe are coupled, thus described by the full action $S$. As a consequence of the compactness of $K$, $\mathfrak{C} (\mathcal{M})$ can be identified with $\mathfrak{A}(\mathcal{M}) \otimes \mathfrak{B}(\mathcal{M})$ outside the causal hull $J^+(K) \cap J^-(K)$ of $K$. To relate the uncoupled algebra with the interacting algebra, i.e., to identify the uncoupled target-probe system with the coupled one at early $(-)$ or late $(+)$ times, one defines respectively the retarded $(-)$ and advanced $(+)$ response maps 
\begin{equation}
\label{eq:response}
\tau^\mp \colon \mathfrak{A}(\mathcal{M}) \otimes \mathfrak{B}(\mathcal{M}) \to \mathcal{C}(\mathcal{M}) \eqend{.}
\end{equation}
In other words, $\tau^\pm$ maps the free quantum field $\Xi$, which is solution of $T \Xi = 0$, into the interacting solution which is free either in the future or the past of the coupling region, respectively. The effect of the interaction can thus be described by the scattering morphism, a map that relates the free algebras in the future and the past of the coupling region. It is given by
\begin{equation}
\label{eq:scattering}
\Theta \colon \mathfrak{A}(\mathcal{M}) \otimes \mathfrak{B}(\mathcal{M}) \to \mathfrak{A}(\mathcal{M}) \otimes \mathfrak{B}(\mathcal{M}) \eqend{,} \quad \Theta \doteq (\tau^-)^{-1} \circ \tau^+ \eqend{,}
\end{equation}
which is indeed an automorphism of the (uncoupled) tensor product algebra constructed out of the future and past response maps $\tau^\pm$.

Based on this, the measurement scheme provides a modeling of the outcome of the actual measurement in terms of the coupled target-probe system $(\Phi,\Psi)$. We assume that both the system $\Phi$ and the probe $\Psi$ are prepared in some initial quantum state at early times, say $\omega$ and $\sigma$, i.e., when the field and the probe are uncoupled. On $\mathfrak{C}(\mathcal{M})$ we define the combined state
\begin{equation}
\label{eq:state-combined}
\tilde{\omega}_\sigma \colon \mathfrak{C}(\mathcal{M}) \to \mathbb{C} \eqend{,} \quad \tilde{\omega}_\sigma \doteq (\omega \otimes \sigma ) \circ (\tau^-)^{-1} \eqend{,}
\end{equation}
which is uncorrelated at early times, and thus describes the correct state of the coupled system prepared in this way. Consider an observable $B \in \mathfrak{B}(\mathcal{M})$ of the probe system, which is measured at late times where it is equal to the element $\1 \otimes B$ in the tensor product algebra. In the coupled system, the corresponding observable is given by $\tilde{B} \doteq \tau^+(\1 \otimes B) \in \mathfrak{C}(\mathcal{M})$. Its expectation value in the state $\tilde{\omega}_\sigma$, that is~\cite[\Eq~(3.8)]{Fewster2018}
\begin{equation}
\label{eq:actual-measure}
\tilde{\omega}_\sigma( \tilde{B} ) = ( \omega \otimes \sigma ) ( \Theta ( \1 \otimes B ) ) \eqend{,}
\end{equation}
describes the expectation of the actual measurement at late times of the target field $\Phi$ through the probe field $\Psi$ under the chosen interaction, which is encoded in the scattering automorphism $\Theta$. As shown in \cite[Sec.~3.2]{Fewster2018}, there exists a unique observable $A \in \mathfrak{A}(\mathcal{M})$ such that $\omega(A)$ coincides with the expectation of the actual measurement given by \Eq~\eqref{eq:actual-measure}. This observable of the target system algebra is called induced system observable and represents the experimental outcome of that measurement scheme. It was further proved \cite[\Eq~(3.14)]{Fewster2018} that the variance of $\tilde{B}$ is always at least as great as the one of the induced observable, namely
\begin{equation}
\operatorname{var}( \tilde{B}; \tilde{\omega}_\sigma ) \doteq \tilde{\omega}_\sigma( \tilde{B}^2 ) - [ \tilde{\omega}_\sigma( \tilde{B} ) ]^2 \geq \operatorname{var}( A; \omega ) = \omega(A^2 - [ \omega(A) ]^2 \eqend{.}
\end{equation} 
This means that the actual measurement is less sharp than the ideal measurement of the induced observable in the target state, due to the quantum fluctuations in the probe which influence the measurement.

\subsection{Induced perturbations and thermal fluctuations}

We shall adopt the algebraic approach to local and covariant measurements for quantum field theories to evaluate the stochastic variance associated to the induced thermal fluctuations studied in Section \ref{sec:thermal-noise}. Our main aim is the proof that the thermal stochastic variance given in \Eq~\eqref{eq:variance-thermal}, which is constructed out of the thermal noise kernel $\tilde{K}_{\mu\nu\alpha\beta}$, corresponds exactly to the variance of the induced fluctuations of the metric obtained in the measurement scheme of Fewster and Verch through the actual measurement \eqref{eq:actual-measure}.

More precisely, we shall prove that this expectation value, and thus its quantum variance, may be evaluated by taking the pair $(\xi^\beta_{\mu\nu},h_{\mu\nu})$ to comprise the system of target field and probe. The thermal stochastic noise $\xi^\beta_{\mu\nu}$ associated to the thermal Maxwell state $\omega_\beta$ given in \Eqs~\eqref{eq:thermal-state}, \eqref{eq:thermal-state-scal-conf} plays the role of a classical (stochastic) target field, while the linearized gravitational perturbation $h_{\mu\nu}$ identifies the probe field. Their coupling~\eqref{eq:action-int} is described by the linear pointwise interaction 
\begin{equation}
\label{eq:action-int-h}
S_\text{int} = - \frac{\kappa}{2} \int \xi^\beta_{\rho\sigma}(x) h^{\rho\sigma}(x) \varrho(x) \total\mu_x \eqend{,} \quad \varrho(x) \in C^\infty_0(\mathcal{M}) \eqend{,}
\end{equation}
which together with the Einstein--Hilbert action for the metric perturbation $h_{\mu\nu}$ comprises the full action in the measurement framework since $\xi^\beta_{\mu\nu}$ is a classical stochastic field.

As already discussed in Section \ref{sec:stochastic}, we are interested in the free intrinsic and induced thermal linearized perturbations $h_{\mu\nu}(x) = h_{\mu\nu}^{0}(x) + h_{\mu\nu}^\text{ind}(x)$, which solve the linearized Einstein--Langevin equation
\begin{equation}
\label{eq:Ein-Lang-therm}
G^{(1)}_{\mu\nu}(x) + \Lambda h_{\mu\nu}(x) = \frac{1}{2} \kappa \xi^\beta_{\mu\nu}(x) \eqend{,}
\end{equation}
where we have again employed a reduction of order argument to remove the tensors $A^{(1)}_{\mu\nu}$ and $B^{(1)}_{\mu\nu}$. To apply the measurement framework of Fewster and Verch, we need to promote the target and the probe fields to quantum fields, thus going beyond the semiclassical and stochastic approximation. In this respect, one views $\kappa$ in \Eq~\eqref{eq:metric-split} as perturbative coupling parameter of the quantum theory, in which one treats the linearized perturbation as a quantum field propagating over the classical background geometry $g^{(0)}_{\mu\nu}(x)$. This new level of approximation may be mathematically studied in the framework of perturbative algebraic quantum field theory, treating local and covariant quantum fields in an effective quantum gravity formulation \cite{Brunetti2013} (see also \cite{Rejzner2016,Froeb_prop2017,Froeb_grav2018,Froeb2023} and references therein). 

On one hand, the Gaussian noise $\xi^\beta_{\mu\nu}$ may be treated as the generator of a $*$-algebra $\mathfrak{A}_\xi(\mathcal{M})$ such that the quasi-free state assigned to $\xi^\beta_{\mu\nu}$ should satisfy \Eq~\eqref{eq:Gaussian-noise}. Namely, one may impose that its one-point function vanishes, and fix its two-point function to the stochastic variance of $\xi^\beta_{\mu\nu}(x)$. For example, in the thermal case studied here the quasi-free state $\omega_\beta \colon \mathfrak{A}_\xi(\mathcal{M}) \to \mathbb{C}$ is uniquely identified by the following relations
\begin{equations}[eq:thermal-state-xi]
\omega_\beta\left( \xi^\beta_{\mu\nu}(x) \right) &= \Estoch \xi^\beta_{\mu\nu}(x) = 0 \eqend{,} \\
\omega_\beta\left( \xi^\beta_{\mu\nu}(x) \xi^\beta_{\rho\sigma}(x') \right) &= \Estoch \xi^\beta_{\mu\nu}(x) \xi^\beta_{\rho\sigma}(x') = \mathcal{N}_{\mu\nu\rho\sigma}(x,x') \eqend{,}
\end{equations}
where $\mathcal{N}_{\mu\nu\rho\sigma}(x,x')$ is given by \Eq~\eqref{eq:noise-kernel-ind} with the kernel $K_{\mu\nu\rho\sigma}$ given by the thermal noise kernel $K^\beta_{\mu\nu\rho\sigma}$ studied in Section \ref{sec:thermal-noise}. In this case, $\mathfrak{A}_\xi(\mathcal{M})$ would be a commutative algebra corresponding to the fact that $\xi^\beta_{\mu\nu}$ is a classical stochastic field which commutes. On the other hand, we may introduce a non-vanishing commutator by replacing $\mathcal{N}_{\mu\nu\rho\sigma}$ with the full unsymmetrized thermal noise kernel $K^\beta_{\mu\nu\rho\sigma}$, which includes an imaginary part. Both approaches ultimately yield the same variance for the observable that we consider, and so we consider the first one.

The quantization of the linearized gravitational field $h_{\mu\nu}$ and the issue of Hadamard states have been already studied in \cite{Ashtekar1982,Fewster2012,Benini2014,Gerard2023}; for applications in inflationary spacetimes, see for example \cite{Froeb2017,Froeb2018,Froeb_prop2017,Froeb_grav2018,Froeb2023}. In the linearized regime, the classical theory in the de Donder gauge $\nabla^\rho h_{\rho\sigma} = 0$ results in the following hyperbolic partial differential equation for $h_{\mu\nu}$:
\begin{equation}
\label{eq:linear-eq-pert}
P_{\mu\nu}{}^{\rho\sigma} h_{\rho\sigma}(x) \doteq \nabla^2_x h_{\mu\nu}(x) - 2 R^\rho{}_{\mu\nu}{}^\sigma h_{\rho\sigma}(x) = \xi^\beta_{\mu\nu}(x) \eqend{,}
\end{equation}
where $\nabla^2_x \doteq g_{(0)}^{\rho\sigma} \nabla_\rho \nabla_\sigma$ is the d'Alembertian of the background spacetime, and $P_{\mu\nu}{}^{\rho\sigma}$ is the normally hyperbolic partial differential operator obtained from the left-hand side of \Eq~\eqref{eq:Ein-Lang-therm}. Notice that the divergence-free and compactness properties of the stochastic noise, which is defined in a compact region, guarantee the solvability of \Eq~\eqref{eq:linear-eq-pert} \cite{Hintz2024}. Thus, the spatially compact retarded and advanced solutions of \Eq~\eqref{eq:linear-eq-pert} are given by 
\begin{equation}
\label{eq:sol-Ein-Lang}
h_{\mu\nu}^{\text{ret}/\text{adv}}(x) = h_{\mu\nu}^0(x) + \frac{1}{2} \kappa \int G^\text{ret/adv}{}_{\mu\nu}{}^{\rho\sigma}(x,x') \xi_{\rho\sigma}^\beta(x') \total\mu_{x'} \eqend{,}
\end{equation}
where $G^\text{ret/adv}{}_{\mu\nu}{}^{\rho\sigma}$ denotes the unique retarded or advanced propagator of $P_{\mu\nu}{}^{\rho\sigma}$. As discussed before, the inhomogeneous part of \Eq~\eqref{eq:sol-Ein-Lang} describes the induced part of the linearized perturbation sourced by the quantum fluctuations.

Based on this and following \cite{Fewster2012}, the quantization of the solutions $h_{\mu\nu}$ may be performed in a similar way as made for the Maxwell field in Section \ref{sec:photon} in the framework of algebraic quantum field theory on curved spacetimes. In this case, the linearized perturbation is a smooth section of the symmetric vector bundle $\mathcal{V} \doteq \operatorname{Sym}(T^*\mathcal{M} \otimes T^*\mathcal{M})$, in particular $h \in \Gamma(\mathcal{V})$ is an element of the configuration space of classical fields which satisfies \Eq~\eqref{eq:linear-eq-pert} with $\xi^\beta_{\mu\nu} = 0$. Thus, the unital $*$-algebra of observables $\mathfrak{A}_h(\mathcal{M})$ on the background spacetime is generated by the smeared fields
\begin{equation}
\label{eq:B-obs}
h(f) \doteq \int_\mathcal{M} h_{\mu\nu} (x) f^{\mu\nu}(x) \total\mu_x \eqend{,}
\end{equation}
for all test tensors $f \in \Gamma_0(\mathcal{V})$. Similarly to the Maxwell field, to take into account the gauge freedom $h_{\mu\nu} \mapsto h_{\mu\nu} + \nabla_{(\mu} \zeta_{\nu)}$ with $\zeta \in C^{\infty}(T(\mathcal{M},\mathbb{R}))$, one defines equivalence classes of fields acting on symmetric, conserved test tensors satisfying $\nabla^{\rho} f_{\rho\sigma} = 0$. Hence, for all $f,g \in \Gamma_0(\mathcal{V})$ of this class the quotient $*$-algebra $\mathfrak{A}_h(\mathcal{M})$ fulfills the following properties:
\begin{enumerate}
\item Linearity: $h(c_0 f + c_1 g) = c_0 h(f) + c_1 h(g)$, $c_0, c_1 \in \mathbb{C}$.
\item Hermiticity: $h^*(f) = h(\Bar{f})$.
\item Symmetry: $h(f) = 0$ for all antisymmetric $f \in \Gamma_0(\mathcal{V})$.
\item On-shell fields: $h(P f) = 0$, where $P$ is defined in \Eq~\eqref{eq:linear-eq-pert}.
\item Canonical commutations relations: $[h(f),h(f')] = \mathi G(f,f') \1$,
\end{enumerate}
where 
\begin{equation}
\label{eq:caus-prop-h}
G(f,f') = \int_{\mathcal{M}\times\mathcal{M}} G_{\mu\nu}{}^{\rho\sigma}(x,x') f^{\mu\nu}(x) f'_{\rho\sigma}(x') \total\mu_x \total\mu_{x'} \eqend{,} \quad G_{\mu\nu}{}^{\rho\sigma} \doteq G^\text{adv}{}_{\mu\nu}{}^{\rho\sigma} - G^\text{ret}{}_{\mu\nu}{}^{\rho\sigma} 
\end{equation}
is the unique causal (advanced minus retarded) Green operator associated to \Eq~\eqref{eq:linear-eq-pert}. To complete the quantization of $h_{\mu\nu}$, we point out that a quantum state, i.e., a linear, positive and normalized functional $\sigma \colon \mathfrak{A}_h(\mathcal{M}) \to \mathbb{C}$ which is quasi-free may be chosen for the linearized perturbation, also in the physical class of Hadamard states.

\medskip

Based on this, the measurement process in our framework proceeds as follows: Let $\mathfrak{A}_\xi(\mathcal{M})$ and $\mathfrak{A}_h(\mathcal{M})$ be respectively the algebras of observables of the target field $\xi(f)$ and the probe field $h(f)$, such that at early times the uncoupled target-probe quantum system is modeled by the tensor algebra $\mathfrak{A}_\xi(\mathcal{M}) \otimes \mathfrak{A}_h (\mathcal{M})$. In the early uncoupled region, a probe observable is of the form $\Xi_0(f) \doteq \1 \otimes h^0(f) \in \mathfrak{A}_\xi(\mathcal{M}) \otimes \mathfrak{A}_h (\mathcal{M})$, where
\begin{equation}
\label{eq:intrinsic-h}
h^0(f) \doteq \int h^0_{\mu\nu}(x) f^{\mu\nu}(x) \total\mu_x \in \Gamma(\mathcal{V})
\end{equation}
is associated to the intrinsic perturbation $h^0_{\mu\nu}(x)$, i.e., a homogeneous solution of \Eq~\eqref{eq:linear-eq-pert}.

Assume that the coupling described in \Eq~\eqref{eq:action-int-h} is localized in a compact region $K$ where the thermal noise $\xi^\beta_{\mu\nu}$ is defined and acts as source in \Eq~\eqref{eq:linear-eq-pert}. Denoting with $\mathfrak{C}_\xi(\mathcal{M})$ the interacting algebra of the coupled target-probe system, we define the retarded $(-)$ and advanced $(+)$ response maps
\begin{equation}
\label{eq:response-xi}
\tau_\xi^\mp \colon \mathfrak{A}_\xi(\mathcal{M}) \otimes \mathfrak{A}_h(\mathcal{M}) \to \mathfrak{C}_\xi(\mathcal{M}) \eqend{,}
\end{equation}
which act on the probe observable as
\begin{equation}
\label{eq:response-xi-act}
\tau_\xi^\mp (\1 \otimes h^0(f)) \doteq \left( \operatorname{id} + \, G^\text{ret/adv} \cdot \xi \otimes \1 \right) (\1 \otimes h^0(f)) = h_{\mu\nu}^{\text{ret}/\text{adv}}(f) \eqend{,}
\end{equation}
where $h_{\mu\nu}^{\text{ret}/\text{adv}}(f)$ are defined by smearing \Eq~\eqref{eq:sol-Ein-Lang} with $f^{\mu\nu}$ as in \Eq~\eqref{eq:B-obs}. The inverse maps are clearly of the form $(\tau_\xi^\mp)^{-1} = \operatorname{id} - \, G^\text{ret/adv} \cdot \xi \otimes \1$, such that the scattering morphism evaluates to
\begin{equation}
\label{eq:scattering-xi}
\Theta_\xi \colon \mathfrak{A}_\xi(\mathcal{M}) \otimes \mathfrak{A}_h(\mathcal{M}) \to \mathfrak{A}_\xi(\mathcal{M}) \otimes \mathfrak{A}_h(\mathcal{M}) \eqend{,} \quad \Theta_\xi \doteq (\tau_\xi^-)^{-1} \cdot \tau_\xi^+ = \operatorname{id} + \, G \cdot \xi \otimes \1 \eqend{,}
\end{equation}
where $G$ denotes the causal propagator \eqref{eq:caus-prop-h}. Therefore, the induced system observable associated to the actual measurement is given by
\begin{equation}
\label{eq:scattering-xi-act}
\tilde{h}(f) \doteq \Theta_\xi \left( \1 \otimes h^0(f) \right) = \1 \otimes h^0(f) + h^\text{ind}(f) \otimes \1 \eqend{,}
\end{equation}
where the intrinsic perturbation $h^0(f)$ was defined in \Eq~\eqref{eq:intrinsic-h}, and
\begin{align}
\label{eq:induced-h}
h^\text{ind}(f) \doteq \frac{1}{2} \kappa \int_{\mathcal{M} \times \mathcal{M}} f^{\mu\nu}(x) G_{\mu\nu\rho\sigma}(x,x') \xi^{\rho\sigma}_\beta(x') \total\mu_x \total\mu_{x'} \in \Gamma(\mathcal{V}) \eqend{,}
\end{align}
denotes the induced perturbation sourced by the thermal stochastic noise.

We may therefore evaluate the variance associated to the induced system observable $\tilde{h}(f)$ as follows:
\begin{theorem}
Consider the algebras $\mathfrak{A}_\xi(\mathcal{M})$ and $\mathfrak{A}_h(\mathcal{M})$ associated to the target field $\xi^\beta_{\mu\nu}$ and the probe field $h_{\mu\nu}$, respectively. Choose for the target field the conformally thermal state $\omega_\beta \colon \mathfrak{A}_\xi(\mathcal{M}) \to \mathbb{C}$, whose one-point and two-point functions are defined in \Eq~\eqref{eq:thermal-state-xi}, and a quasi-free quantum state $\sigma \colon \mathfrak{A}_h(\mathcal{M}) \to \mathbb{C}$ for the probe. 

Denoting with $\omega^\beta_\sigma \doteq \omega_\beta \otimes \sigma$ the combined state of the target-probe system on $\mathfrak{A}_\xi(\mathcal{M}) \otimes \mathfrak{A}_h(\mathcal{M})$, and recalling the definition of the induced observable $\tilde{h}(f)$ in \Eq~\eqref{eq:scattering-xi-act}, we have
\begin{equation}
\operatorname{var} \left( \tilde{h}(f); \omega^\beta_\sigma \right) = \operatorname{var} \left( h^0(f); \sigma \right) + \operatorname{var}( h^\mathrm{ind}(f); \omega_\beta ) \eqend{.}
\end{equation}

Moreover, given the causal propagator of the probe \eqref{eq:caus-prop-h} and the thermal noise kernel $K^\beta_{\rho\sigma\gamma\delta}(x,y)$ studied in Section \ref{sec:thermal-noise}, define the gauge-independent, causal thermal noise kernel
\begin{equation}
\label{eq:thermal-noise-caus}
\mathcal{K}^\beta_{\mu\nu\alpha\beta} (x,y) \doteq \int_{\mathcal{M} \times \mathcal{M}} G_{\mu\nu}{}^{\rho\sigma}(x,x') G_{\alpha\beta}{}^{\gamma\delta}(y,y') K^\beta_{\rho\sigma\gamma\delta} (x',y') \total\mu_{x'} \total\mu_{y'} \eqend{.}
\end{equation}
Then the thermal variance of the induced perturbation $h^\mathrm{ind}_{\mu\nu}(x)$ induced by the stochastic noise $\xi_{\mu\nu}(x)$ in the state $\omega_\beta$ is given by
\begin{splitequation}
\label{eq:thermal-noise-varind}
\operatorname{var}( h^\mathrm{ind}(f); \omega_\beta) &\doteq \omega_\beta \left[ \left( h^\mathrm{ind}(f) \right)^2 \right] - \left[ \omega_\beta\left( h^\mathrm{ind}(f) \right) \right]^2 \\
&= \frac{\kappa^2}{4} \int_{\mathcal{M} \times \mathcal{M}} \mathcal{K}^\beta_{\mu\nu\alpha\beta}(x,y) f^{\mu\nu}(x) f^{\alpha\beta} (y) \total\mu_x \total\mu_y \eqend{.}
\end{splitequation}
\end{theorem}

\begin{proof}
Consider the combined state $\omega^\beta_\sigma \colon \mathfrak{C}_\xi(\mathcal{M}) \to \mathbb{C}$ seen as a functional on the interacting algebra of the target-probe system, which has thus the form
\begin{equation}
\omega^\beta_\sigma \doteq (\tau_\xi^-)^{-1*} ( \omega_\beta \otimes \sigma ) \eqend{,}
\end{equation}
with $\tau_\xi^- = \operatorname{id} - \ G^\text{ret} \cdot \xi \otimes \1$ defined as before. In a similar way as in \cite[Sec.~5.1]{Fewster2018}, using the expression of the scattering morphism \eqref{eq:scattering-xi-act}, and recalling that the one-point functions associated to both $\omega_\beta$ and $\sigma$ vanish, we evaluate
\begin{equation}
\begin{aligned}
\operatorname{var}( \tilde{h}(f); \omega^\beta_\sigma ) &= \omega^\beta_\sigma\left( \tilde{h}(f)^2 \right) - \left[ \omega^\beta_\sigma\left( \tilde{h}(f) \right) \right]^2 \\
&= ( \omega_\beta \otimes \sigma ) \left( \left[ \Theta_\xi ( \1 \otimes h^0(f) ) \right]^2 \right) - \left[ ( \omega_\beta \otimes \sigma ) \left( \Theta_\xi ( \1 \otimes h^0(f) ) \right) \right]^2 \\
&= ( \omega_\beta \otimes \sigma ) \left[ \left( \1 \otimes h^0(f) + h^{\text{ind}}(f) \otimes \1 \right)^2 \right] \\
&\quad- \left[ ( \omega_\beta \otimes \sigma ) \left( \1 \otimes h^0(f)^2 + h^{\text{ind}}(f)^2 \otimes \1 \right) \right]^2 \\
&= \sigma\left( h^0(f)^2 \right) - \sigma\left( h^0(f) \right)^2 + \omega_\beta\left( h^\text{ind}(f)^2 \right) - \omega_\beta\left( h^\text{ind}(f) \right)^2 \\
&= \operatorname{var}( h^0(f); \sigma ) + \operatorname{var}( h^\text{ind}(f); \omega_\beta ) \eqend{.}
\end{aligned}
\end{equation}
Hence, it is straightforward to evaluate the variance of the induced fluctuations, which is obtained by applying the scattering morphism to each argument of the noise kernel $K^\beta_{\rho\sigma\gamma\delta}(x,x')$, and smear it with the test tensor associated to the probe observable $h^0(f)$. Namely, using \Eq~\eqref{eq:induced-h} and recalling that $\omega_\beta\left( h^\text{ind}(f) \right) = 0$, we obtain
\begin{equation}
\begin{aligned}
\operatorname{var}( h^\text{ind}(f);\omega_\beta) &= \omega_\beta\left( \left[ h^\text{ind}(f) \right]^2 \right) - \left[ \omega_\beta\left( h^\text{ind}(f) \right) \right]^2 \\
&= \frac{1}{4} \kappa^2 \int_{\mathcal{M} \times \mathcal{M}} \total\mu_x \total\mu_y f^{\mu\nu}(x) f^{\alpha\beta}(y) \\
&\qquad\times \int_{\mathcal{M} \times \mathcal{M}} \total\mu_{x'} \total\mu_{y'} G_{\mu\nu}{}^{\rho\sigma}(x,x') G_{\alpha\beta}{}^{\gamma\delta}(y,y') K^\beta_{\rho\sigma\gamma\delta} (x',y') \\
&= \frac{\kappa^2}{4} \int_{\mathcal{M} \times \mathcal{M}} \mathcal{K}^\beta_{\mu\nu\alpha\beta}(x,y) f^{\mu\nu}(x) f^{\alpha \beta} (y) \total\mu_x \total\mu_y \eqend{,} 
\end{aligned}
\end{equation}
thus concluding the proof.
\end{proof}
\begin{remark}
Since the kernel $\mathcal{K}^\beta_{\mu\nu\alpha\beta}$ in the variance~\eqref{eq:thermal-noise-varind} is smeared with the same test tensor $f^{\mu\nu}$ in both arguments, only its symmetric part enters. In turn, this means that only the symmetric part of $K^\beta_{\rho\sigma\gamma\delta}$ contributes in \Eq~\eqref{eq:thermal-noise-caus}, and hence we can replace $K^\beta_{\rho\sigma\gamma\delta}$ by $\mathcal{N}_{\rho\sigma\gamma\delta}$~\eqref{eq:noise-kernel-ind}. Therefore, our result~\eqref{eq:thermal-noise-varind} for the thermal variance is independent of whether we consider the algebra $\mathfrak{A}_\xi(\mathcal{M})$ of the target field to be commutative or non-commutative.
\end{remark}

\section{Conclusions and future outlook}
\label{sec:conc}

In this work, we studied the quantum fluctuations of the metric driven by a thermal bath of conformally invariant Maxwell fields on cosmological spacetimes in the framework of stochastic gravity. We studied carefully the thermal noise kernel associated with the Gaussian stochastic noise entering the Einstein--Langevin equation, proving that it is related to the variance of the induced fluctuations of the linearized metric perturbation which are sourced by the stochastic noise. Our result is obtained in the local measurement framework of Fewster and Verch, and it may be extended to other classes of quantum matter fields and background spacetimes. For example, our evaluations may be straightforward applied to Minkowski spacetime and KMS quantum states in flat spacetime, evaluating the linearized Einstein--Langevin equation in the limit $a(\eta) \to 1$. This limit applied to the thermal noise kernel provides the quantum fluctuations driven by thermal photons on flat spacetime.

On the other hand, unlike the zeroth-order background case, our analysis does not focus its attention on studying solutions of the linearized semiclassical and stochastic equations. It is well-known in the literature that unstable runaway solutions may appear in these semiclassical and stochastic equations, because of which the validity of those equations is still debated, see for example \cite{Simon1991,Hu2001,Anderson2003}. It is expected that this issue should be addressed in the future, in order to completely understand the nature of the induced metric perturbations driven by the backreaction of quantum fields and their fluctuations. We removed these runaway solutions by employing a reduction of order approach, which allowed us to ignore the higher-order derivative curvature tensors appearing in \Eq~\eqref{eq:Ein-Lang-lin}, which in principle might spoil the hyperbolic nature of the equation. A more carefully analysis of this issue should be performed in the future to understand possible generalizations to more complicated linearized solutions. Such a more in-depth analysis should also take into account the linearized expectation value of the quantum stress-energy tensor, which sources the intrinsic perturbations $h^{\text{intr},1}_{\mu\nu}$ that are modulated by the backreaction of the quantum matter field. 

It would also be interesting to study more complicated measurement processes, in which other types of target and probe fields may be involved, or different target-probe couplings may be taken into account beyond the linear point-wise case \eqref{eq:action-int-h}. For instance, one might consider spinor fields, non-conformally invariant quantum fields, non-Gaussian stochastic noises, or quadratic interactions.

\section*{Acknowledgements}

The work of PM was financed by the European Union -- NextGenerationEU -- National Recovery and Resilience Plan (NRRP) -- Mission 4 Component 2 Investment 1.2 –- ``Funding projects presented by young researchers'' Seal of Excellence PNRR Young Researchers, ``SPACE project'' -- SOE20240000129 -- CUP E63C24002410003. He thanks E. D'Angelo and N. Pinamonti for useful discussions. MBF acknowledges financial support from the COST action CA23130 -- Bridging high and low energies in search of quantum gravity (BridgeQG) through the STSM ``Thermal fluctuations in quantum gravity: noise kernel and stochastic gravity''. He thanks Paolo Meda, the Mathematical Physics group and the Department of Mathematics of the University of Trento for their generous hospitality. DG was supported by project 24-13079S of the Czech Science Foundation (GA{\v C}R).

\appendix

\section{Derivations of thermal noise kernel}
\label{sec:app}

In a first step, it is possible to reduce the full kernel $\tilde{K}^\beta_{\mu\nu\rho\sigma}(\eta-\eta',\vec{p})$ to just three independent terms. Namely, using the symmetry~\eqref{eq:induced_symmetric}, the tracelessness~\eqref{eq:induced_traceless} and the conservation~\eqref{eq:induced_conserved}, a long but straightforward computation shows that we have
\begin{equations}[eq:appendix_kernel_components]
\tilde{K}^\beta_{0000}(\eta-\eta',\vec{p}) &= \tilde{L}_\text{A}(\eta-\eta',\vec{p}) \eqend{,} \\
\tilde{K}^\beta_{000i}(\eta-\eta',\vec{p}) &= - \mathi \frac{\vec{p}_i}{2 \vec{p}^2} \partial_\eta \tilde{L}_\text{B}(\eta-\eta',\vec{p}) \eqend{,} \\
\begin{split}
\tilde{K}^\beta_{00ij}(\eta-\eta',\vec{p}) &= \frac{\eta_{ij}}{4 \vec{p}^2} \left[ - \left( \partial_\eta^2 - \vec{p}^2 \right) \tilde{L}_\text{B}(\eta-\eta',\vec{p}) + \partial_\eta^2 \tilde{L}_\text{C}(\eta-\eta',\vec{p}) \right] \\
&\quad+ \frac{\vec{p}_i \vec{p}_j}{4 \abs{\vec{p}}^4} \Bigl[ 4 \vec{p}^2 \tilde{L}_\text{A}(\eta-\eta',\vec{p}) + 3 \left( \partial_\eta^2 - \vec{p}^2 \right) \tilde{L}_\text{B}(\eta-\eta',\vec{p}) \\
&\hspace{4em}- 3 \partial_\eta^2 \tilde{L}_\text{C}(\eta-\eta',\vec{p}) \Bigr] \eqend{,}
\end{split} \\
\begin{split}
\tilde{K}^\beta_{0i0j}(\eta-\eta',\vec{p}) &= \frac{\eta_{ij}}{4 \vec{p}^2} \left( \partial_\eta^2 + \vec{p}^2 \right) \tilde{L}_\text{B}(\eta-\eta',\vec{p}) \\
&\quad- \frac{\vec{p}_i \vec{p}_j}{4 \abs{\vec{p}}^4} \left[ 4 \vec{p}^2 \tilde{L}_\text{A}(\eta-\eta',\vec{p}) + \left( 3 \partial_\eta^2 - \vec{p}^2 \right) \tilde{L}_\text{B}(\eta-\eta',\vec{p}) \right] \eqend{,}
\end{split} \\
\begin{split}
\tilde{K}^\beta_{0ikl}(\eta-\eta',\vec{p}) &= - \mathi \frac{\vec{p}_{(k} \eta_{l)i}}{4 \abs{\vec{p}}^4} \partial_\eta \left( \partial_\eta^2 + \vec{p}^2 \right) \tilde{L}_\text{C}(\eta-\eta',\vec{p}) \\
&\quad- \mathi \frac{\vec{p}_i \eta_{kl}}{8 \abs{\vec{p}}^4} \partial_\eta \left[ 4 \vec{p}^2 \tilde{L}_\text{B}(\eta-\eta',\vec{p}) + \left( \partial_\eta^2 - \vec{p}^2 \right) \tilde{L}_\text{C}(\eta-\eta',\vec{p}) \right] \\
&\quad+ \mathi \frac{\vec{p}_i \vec{p}_k \vec{p}_l}{8 \abs{\vec{p}}^6} \partial_\eta \left[ 8 \vec{p}^2 \tilde{L}_\text{B}(\eta-\eta',\vec{p}) + \left( 5 \partial_\eta^2 - \vec{p}^2 \right) \tilde{L}_\text{C}(\eta-\eta',\vec{p}) \right] \eqend{,}
\end{split} \\
\begin{split}
\tilde{K}^\beta_{ijkl}(\eta-\eta',\vec{p}) &= \frac{\vec{p}_i \vec{p}_j \vec{p}_k \vec{p}_l}{32 \abs{\vec{p}}^8} \biggl[ 32 \abs{\vec{p}}^4 \tilde{L}_\text{A}(\eta-\eta',\vec{p}) + 16 \vec{p}^2 \left( 5 \partial_\eta^2 - \vec{p}^2 \right) \tilde{L}_\text{B}(\eta-\eta',\vec{p}) \\
&\hspace{6em}+ \left( 35 \partial_\eta^4 - 10 \vec{p}^2 \partial_\eta^2 + 3 \abs{\vec{p}}^4 \right) \tilde{L}_\text{C}(\eta-\eta',\vec{p}) \biggr] \\
&\quad- \frac{\vec{p}_i \vec{p}_j \eta_{kl} + \vec{p}_k \vec{p}_l \eta_{ij}}{32 \abs{\vec{p}}^6} \biggl[ 8 \vec{p}^2 \left( \partial_\eta^2 - \vec{p}^2 \right) \tilde{L}_\text{B}(\eta-\eta',\vec{p}) \\
&\hspace{11em}+ \left( 5 \partial_\eta^4 + 2 \vec{p}^2 \partial_\eta^2 + 5 \abs{\vec{p}}^4 \right) \tilde{L}_\text{C}(\eta-\eta',\vec{p}) \biggr] \\
&\quad- \frac{\vec{p}_{(i} \eta_{j)(k} \vec{p}_{l)}}{8 \abs{\vec{p}}^6} \biggl[ 8 \vec{p}^2 \left( \partial_\eta^2 + \vec{p}^2 \right) \tilde{L}_\text{B}(\eta-\eta',\vec{p}) \\
&\hspace{4em}+ \left( 5 \partial_\eta^2 - 3 \vec{p}^2 \right) \left( \partial_\eta^2 + \vec{p}^2 \right) \tilde{L}_\text{C}(\eta-\eta',\vec{p}) \biggr] \\
&\quad+ \frac{\eta_{ij} \eta_{kl} + 2 \eta_{i(k} \eta_{l)j}}{32 \abs{\vec{p}}^4} \left( \partial_\eta^2 + \vec{p}^2 \right)^2 \tilde{L}_\text{C}(\eta-\eta',\vec{p}) \eqend{,}
\end{split}
\end{equations}
where we defined the combinations
\begin{equations}
\tilde{L}_\text{A}(\tau,\vec{p}) &\doteq \tilde{K}^\beta_{0000}(\tau,\vec{p}) \eqend{,} \\
\tilde{L}_\text{B}(\tau,\vec{p}) &\doteq \tilde{K}^\beta_{0000}(\tau,\vec{p}) + \eta^{ij} \tilde{K}^\beta_{0i0j}(\eta-\eta',\vec{p}) \eqend{,} \\
\tilde{L}_\text{C}(\tau,\vec{p}) &\doteq \tilde{K}^\beta_{0000}(\tau,\vec{p}) + 2 \eta^{ij} \tilde{K}^\beta_{0i0j}(\tau,\vec{p}) + \eta^{ik} \eta^{jl} \tilde{K}^\beta_{ijkl}(\tau,\vec{p}) \eqend{.}
\end{equations}

For the individual terms, we have
\begin{splitequation}
\label{eq:appendix_ka}
\tilde{K}^\beta_{0000}(\tau,\vec{p}) &= \lim_{\epsilon \to 0^+} \int \frac{\abs{\vec{p}-\vec{q}} \abs{\vec{q}}}{4} \left[ \frac{\mathe^{\mathi \abs{\vec{p}-\vec{q}} (\tau+\mathi \epsilon)}}{1 - \mathe^{- \beta \abs{\vec{p}-\vec{q}}}} + \frac{\mathe^{- \mathi \abs{\vec{p}-\vec{q}} (\tau-\mathi \epsilon)}}{\mathe^{\beta \abs{\vec{p}-\vec{q}}} - 1} \right] \\
&\hspace{8em}\times \left[ \frac{\mathe^{\mathi \abs{\vec{q}} (\tau+\mathi \epsilon)}}{1 - \mathe^{- \beta \abs{\vec{q}}}} + \frac{\mathe^{- \mathi \abs{\vec{q}} (\tau-\mathi \epsilon)}}{\mathe^{\beta \abs{\vec{q}}} - 1} \right] \frac{\total^3 \vec{q}}{(2 \pi)^3} \eqend{,}
\end{splitequation}
\begin{splitequation}
\label{eq:appendix_kb}
\eta^{ij} \tilde{K}^\beta_{0i0j}(\tau,\vec{p}) &= \lim_{\epsilon \to 0^+} \int \frac{\vec{p}^2 - \vec{q}^2 - \abs{\vec{p}-\vec{q}}^2}{8} \left[ \frac{\mathe^{\mathi \abs{\vec{p}-\vec{q}} (\tau+\mathi \epsilon)}}{1 - \mathe^{- \beta \abs{\vec{p}-\vec{q}}}} - \frac{\mathe^{- \mathi \abs{\vec{p}-\vec{q}} (\tau-\mathi \epsilon)}}{\mathe^{\beta \abs{\vec{p}-\vec{q}}} - 1} \right] \\
&\hspace{8em}\times \left[ \frac{\mathe^{\mathi \abs{\vec{q}} (\tau+\mathi \epsilon)}}{1 - \mathe^{- \beta \abs{\vec{q}}}} - \frac{\mathe^{- \mathi \abs{\vec{q}} (\tau-\mathi \epsilon)}}{\mathe^{\beta \abs{\vec{q}}} - 1} \right] \frac{\total^3 \vec{q}}{(2 \pi)^3} \eqend{,}
\end{splitequation}
and
\begin{splitequation}
\label{eq:appendix_kc}
\eta^{ik} \eta^{jl} \tilde{K}^\beta_{ijkl}(\tau,\vec{p}) &= \lim_{\epsilon \to 0^+} \int \frac{\left( \vec{p}^2 - \vec{q}^2 - \abs{\vec{p}-\vec{q}}^2 \right)^2}{16 \abs{\vec{p}-\vec{q}} \abs{\vec{q}}} \left[ \frac{\mathe^{\mathi \abs{\vec{p}-\vec{q}} (\tau+\mathi \epsilon)}}{1 - \mathe^{- \beta \abs{\vec{p}-\vec{q}}}} + \frac{\mathe^{- \mathi \abs{\vec{p}-\vec{q}} (\tau-\mathi \epsilon)}}{\mathe^{\beta \abs{\vec{p}-\vec{q}}} - 1} \right] \\
&\hspace{8em}\times \left[ \frac{\mathe^{\mathi \abs{\vec{q}} (\tau+\mathi \epsilon)}}{1 - \mathe^{- \beta \abs{\vec{q}}}} + \frac{\mathe^{- \mathi \abs{\vec{q}} (\tau-\mathi \epsilon)}}{\mathe^{\beta \abs{\vec{q}}} - 1} \right] \frac{\total^3 \vec{q}}{(2 \pi)^3} \eqend{,}
\end{splitequation}
where we used that $2 (\vec{p} \cdot \vec{q}) = \vec{p}^2 + \vec{q}^2 - \abs{\vec{p}-\vec{q}}^2$.

These kernels can be evaluated using spherical coordinates with $(\vec{p} \cdot \vec{q}) = \abs{\vec{p}} \abs{\vec{q}} \cos \theta$ and $\total^3 \vec{q} = \abs{\vec{q}}^2 \sin \theta \total \abs{\vec{q}} \total \theta \total \phi$, but the resulting integrals are too complicated to compute exactly. Hence, in the following subsections, we compute them in various limits.

\subsection{Small external momentum limit}
\label{app:smallp}

The limit of small external momentum $\vec{p} \to 0$ is simple to take, since the integrals converge for any value of $\vec{p}$ and we may exchange the limit with the integration over $\vec{q}$. To obtain the expansion of the components~\eqref{eq:appendix_kernel_components} of the kernel $\tilde{K}_{\mu\nu\rho\sigma}$ to leading in $\vec{p}$, we see that we need the expansion of $\tilde{L}_\text{A}$ to zeroth order in $\vec{p}$, the expansion of $\tilde{L}_\text{B}$ to second order in $\vec{p}$, and the expansion of $\tilde{L}_\text{C}$ to fourth order in $\vec{p}$. We note that for any $\beta > 0$, the $\mathi \epsilon$ prescription in the second term of each bracket in each of the integrals~\eqref{eq:appendix_ka}, \eqref{eq:appendix_kb} and~\eqref{eq:appendix_kc} is irrelevant because the exponential in the denominator already ensures the decay of the integrand for large $\abs{\vec{q}}$, and so we may change the prescription to have $\tau + \mathi \epsilon$ in all terms.

Expanding the kernels for small $\vec{p}$ and performing the easy integrals over $\theta$ and $\phi$, we then obtain
\begin{equations}[eq:appendix_smallp_labc]
\tilde{L}_\text{A}(\tau,\vec{p}) &= \tilde{L}_\text{A}^{(0)}(\tau) + \bigo{\vec{p}^2} \eqend{,} \\
\tilde{L}_\text{B}(\tau,\vec{p}) &= \tilde{L}_\text{B}^{(0)}(\tau) + \tilde{L}_\text{B}^{(1)}(\tau) \vec{p}^2 + \bigo{\abs{\vec{p}}^4} \eqend{,} \\
\tilde{L}_\text{C}(\tau,\vec{p}) &= 2 \tilde{L}_\text{B}^{(0)}(\tau) + \tilde{L}_\text{C}^{(1)}(\tau) \vec{p}^2 + \tilde{L}_\text{C}^{(2)}(\tau) \abs{\vec{p}}^4 + \bigo{\abs{\vec{p}}^6}
\end{equations}
with
\begin{equations}[eq:appendix_smallp_labc_components]
\tilde{L}_\text{A}^{(0)}(\tau) &= \frac{1}{8 \pi^2} I_4(\tau) \eqend{,} \\
\tilde{L}_\text{B}^{(0)}(\tau) &= \frac{J_{1,4} + J_{2,4}}{2 \pi^2} \eqend{,} \\
\begin{split}
\tilde{L}_\text{B}^{(1)}(\tau) &= \frac{1}{24 \pi^2} I_2(\tau) + \frac{\beta^2 + 2 \mathi \beta \tau - 2 \tau^2}{24 \pi^2} \left( J_{1,4} + J_{2,4} \right) \\
&\quad+ \frac{\beta^2}{6 \pi^2} \left( J_{2,4} + 2 J_{3,4} + J_{4,4} \right) - \frac{\beta}{6 \pi^2} \left( J_{1,3} + 3 J_{2,3} + 2 J_{3,3} \right) \eqend{,}
\end{split} \\
\begin{split}
\tilde{L}_\text{C}^{(1)}(\tau) &= - \frac{\beta}{3 \pi^2} \left( J_{1,3} + 3 J_{2,3} + 2 J_{3,3} \right) \\
&\quad+ \frac{\beta^2 + 2 \mathi \beta \tau - 2 \tau^2}{12 \pi^2} \left( J_{1,4} + J_{2,4} \right) + \frac{\beta^2}{3 \pi^2} \left( J_{2,4} + 2 J_{3,4} + J_{4,4} \right) \eqend{,}
\end{split} \\
\begin{split}
\label{eq:appendix_smallp_labc_components_lc2}
\tilde{L}_\text{C}^{(2)}(\tau) &= \frac{(\beta + \mathi \tau)^4 + \tau^4}{240 \pi^2} \left( J_{1,4} + J_{2,4} \right) + \frac{\beta^2 \left( 2 \beta^2 + 3 \mathi \beta \tau - 3 \tau^2 \right)}{30 \pi^2} \left( J_{2,4} + 2 J_{3,4} + J_{4,4} \right) \\
&\quad+ \frac{\beta^4}{5 \pi^2} \left( J_{3,4} + 3 J_{4,4} + 3 J_{5,4} + J_{6,4} \right) - \frac{\beta \left( \beta^2 + 3 \mathi \beta \tau - 3 \tau^2 \right)}{30 \pi^2} \left( J_{1,3} + 3 J_{2,3} + 2 J_{3,3} \right) \\
&\quad- \frac{\beta^3}{5 \pi^2} \left( J_{2,3} + 4 J_{3,3} + 5 J_{4,3} + 2 J_{5,3} \right) + \frac{\beta^2 + 2 \mathi \beta \tau - 2 \tau^2}{30 \pi^2} \left( J_{1,2} + J_{2,2} \right) \\
&\quad+ \frac{2 \beta^2}{15 \pi^2} \left( J_{2,2} + 2 J_{3,2} + J_{4,2} \right) + \frac{\beta}{30 \pi^2} \left( J_{1,1} + 3 J_{2,1} + 2 J_{3,1} \right) + \frac{1}{60 \pi^2} I_0(\tau) \eqend{.}
\end{split}
\end{equations}
Here, we defined the integrals
\begin{equation}
\label{eq:appendix_in_def}
I_n(\tau) \doteq \lim_{\epsilon \to 0^+} \lim_{\mu \to 0^+} \left[ \int_\mu^\infty \left[ \frac{\mathe^{\mathi q (\tau + \mathi \epsilon)}}{1 - \mathe^{- \beta q}} + \frac{\mathe^{- \mathi q (\tau + \mathi \epsilon)}}{\mathe^{\beta q} - 1} \right]^2 q^n \total q + c_n(\mu) \right]
\end{equation}
and
\begin{equation}
\label{eq:appendix_jmn_def}
J_{m,n} \doteq \lim_{\mu \to 0^+} \left[ \int_\mu^\infty \frac{q^n}{\left( \mathe^{\beta q} - 1 \right)^m} \total q + d_{m,n}(\mu) \right] \eqend{,}
\end{equation}
where
\begin{equation}
c_0(\mu) = - \frac{4}{\beta^2 \mu} \eqend{,} \quad c_1(\mu) = \frac{4}{\beta^2} \ln \mu \eqend{,}
\end{equation}
and $d_{m,n}(\mu)$ are linear combinations of negative powers of $\mu$ and $\ln \mu$ such that the limit $\mu \to 0$ is finite. Note that the explicit expressions for $d_{m,n}$ are irrelevant since they cancel out in the sum~\eqref{eq:appendix_smallp_labc_components_lc2} for $\tilde{L}_\text{C}^{(2)}$.

It remains to actually compute the integrals $I_n$ and $J_{m,n}$. We start with the latter ones, for which we have
\begin{lemma}
\label{lemma:integral_jmn}
For $n \geq m \geq 1$, the integrals $J_{m,n}$~\eqref{eq:appendix_jmn_def} fulfill the recursion relation
\begin{equation}
J_{m+1,n+1} = \frac{n+1}{m \beta} J_{m,n} - J_{m,n+1} \eqend{.}
\end{equation}
For $n \geq m = 0$, we have instead the result
\begin{equation}
J_{1,n+1} = \frac{\Gamma(n+2) \zeta(n+2)}{\beta^{n+2}} \eqend{.}
\end{equation}
\end{lemma}
\begin{proof}
We compute
\begin{equation}
J_{m+1,n+1} - \frac{n+1}{m \beta} J_{m,n+1-1} + J_{m,n+1} = - \frac{1}{m \beta} \int_0^\infty \partial_q \left[ \frac{q^{n+1}}{\left( \mathe^{\beta q} - 1 \right)^m} \right] \total q \eqend{,}
\end{equation}
and the boundary terms vanish under the stated assumption. For $m = 0$, we rescale the integration variable $q = x/\beta$ and then employ the integral representation~\cite[\Eq~(25.5.1)]{dlmf} for the Riemann $\zeta$ function.
\end{proof}
Explicitly, we need
\begin{equation}
\label{eq:appendix_smallp_j1424}
J_{1,4} + J_{2,4} = \frac{4}{\beta} J_{1,3} = \frac{4 \pi^4}{15 \beta^5} \eqend{.}
\end{equation}

For the former integrals, we have instead
\begin{lemma}
\label{lemma:integral_in}
The integrals $I_n$~\eqref{eq:appendix_in_def} with $n \geq 2$ satisfy the recursion relation
\begin{equation}
I_n(\tau) = - \frac{1}{4} \partial_\tau^2 I_{n-2}(\tau) + 2 J_{1,n} + 2 J_{2,n} \eqend{.}
\end{equation}
For $n = 0,1$ instead, only derivatives are well-defined and we have for $k \geq 1$
\begin{equation}
\label{eq:appendix_smallp_dki0}
\partial_\tau^k I_0(\tau) = \frac{\mathi \pi}{\beta^2} \lim_{\epsilon \to 0^+} \partial_\tau^k \left[ \bigl[ \beta + 2 \mathi (\tau + \mathi \epsilon) \bigr] \coth\left( \frac{2 \pi (\tau + \mathi \epsilon)}{\beta} \right) \right] \eqend{.}
\end{equation}
(We do not need the integral $I_1$.)
\end{lemma}
\begin{proof}
We compute for $n \geq 4$ that
\begin{splitequation}
I_n(\tau) - 2 J_{1,n} - 2 J_{2,n} &= \lim_{\epsilon \to 0^+} \int_0^\infty \left[ \frac{\mathe^{2 \mathi q (\tau + \mathi \epsilon)}}{\left( 1 - \mathe^{- \beta q} \right)^2} + \frac{\mathe^{- 2 \mathi q (\tau + \mathi \epsilon)}}{\left( \mathe^{\beta q} - 1 \right)^2} \right] q^n \total q \\
&= - \frac{1}{4} \lim_{\epsilon \to 0^+} \int_0^\infty \partial_\tau^2 \left[ \frac{\mathe^{2 \mathi q (\tau + \mathi \epsilon)}}{\left( 1 - \mathe^{- \beta q} \right)^2} + \frac{\mathe^{- 2 \mathi q (\tau + \mathi \epsilon)}}{\left( \mathe^{\beta q} - 1 \right)^2} \right] q^{n-2} \total q \eqend{,}
\end{splitequation}
and since the integral is absolutely convergent for every $\epsilon > 0$, we can interchange the derivatives with the integration to obtain
\begin{splitequation}
I_n(\tau) - 2 J_{1,n} - 2 J_{2,n} = - \frac{1}{4} \partial_\tau^2 \left[ I_{n-2}(\tau) - 2 J_{1,n-2} - 2 J_{2,n-2} \right] = - \frac{1}{4} \partial_\tau^2 I_{n-2}(\tau) \eqend{,}
\end{splitequation}
since the $J_{m,n}$ do not depend on $\tau$. For $n = 2,3$, we instead have to put the explicit infrared cutoff $\mu$ before we can interchange the derivatives with the integration, and obtain
\begin{splitequation}
I_n(\tau) - 2 J_{1,n} - 2 J_{2,n} &= - \frac{1}{4} \lim_{\epsilon \to 0^+} \lim_{\mu \to 0^+} \partial_\tau^2 \int_\mu^\infty \left[ \frac{\mathe^{2 \mathi q (\tau + \mathi \epsilon)}}{\left( 1 - \mathe^{- \beta q} \right)^2} + \frac{\mathe^{- 2 \mathi q (\tau + \mathi \epsilon)}}{\left( \mathe^{\beta q} - 1 \right)^2} \right] q^{n-2} \total q \\
&= - \frac{1}{4} \partial_\tau^2 I_{n-2}(\tau) \eqend{,}
\end{splitequation}
since the constants $c_n(\mu)$ needed to define $I_0(\tau)$ and $I_1(\tau)$ are also independent of $\tau$.

For the explicit result for $I_0$, we compute
\begin{splitequation}
\label{eq:appendix_smallp_dtau_i0}
\partial_\tau I_0(\tau) &= 2 \mathi \lim_{\epsilon \to 0^+} \int_0^\infty \left[ \frac{\mathe^{2 \mathi q (\tau + \mathi \epsilon)}}{\left( 1 - \mathe^{- \beta q} \right)^2} - \frac{\mathe^{- 2 \mathi q (\tau + \mathi \epsilon)}}{\left( \mathe^{\beta q} - 1 \right)^2} \right] q \total q \\
&= - \frac{2 \mathi}{\beta^2} \lim_{\epsilon \to 0^+} \int_0^1 \left[ z^{- \frac{2 \mathi (\tau + \mathi \epsilon)}{\beta}} - z^2 z^\frac{2 \mathi (\tau + \mathi \epsilon)}{\beta} \right] \frac{\ln z}{z (1-z)^2} \total z
\end{splitequation}
using the change of variables $q = - \frac{1}{\beta} \ln z$. For $\epsilon > 0$, the integrand is an integrable function both at $z = 0$ and $z = 1$, and thus the integral is well-defined. To determine its value, we employ the integral representation~\cite[\Eq~(15.6.1)]{dlmf} for the hypergeometric function
\begin{equation}
\label{eq:appendix_smallp_hypergeom_int}
\int_0^1 \frac{t^{b-1} (1-t)^{c-b-1}}{(1-z t)^a} \total t = \frac{\Gamma(b) \Gamma(c-b)}{\Gamma(c)} \hypergeom{2}{1}\left( a, b; c; z \right) \eqend{,}
\end{equation}
which holds for $\Re c > \Re b > 0$ and $z \not\in [1,\infty)$. Namely, shifting $c \to c+a+b$, taking derivatives of this equation with respect to $b$ and taking the limit as $z \to 1$, we obtain
\begin{equation}
\label{eq:appendix_smallp_hypergeom_intlog}
\int_0^1 t^{b-1} (1-t)^{c-1} \ln^k t \total t = \lim_{z \to 1^-} \frac{\partial^k}{\partial b^k} \left[ \frac{\Gamma(b) \Gamma(a+c)}{\Gamma(a+b+c)} \hypergeom{2}{1}\left( a, b; a+b+c; z \right) \right] \eqend{,}
\end{equation}
and thus (choosing $a = 2$ and $c = -1$)
\begin{splitequation}
\partial_\tau I_0(\tau) &= - \frac{2 \mathi}{\beta^2} \lim_{\epsilon \to 0^+} \lim_{z \to 1^-} \frac{\partial}{\partial b} \Biggl[ \frac{1}{\left( b - \frac{2 \mathi (\tau + \mathi \epsilon)}{\beta} \right)} \hypergeom{2}{1}\left( 2, b - \frac{2 \mathi (\tau + \mathi \epsilon)}{\beta}; b + 1 - \frac{2 \mathi (\tau + \mathi \epsilon)}{\beta}; z \right) \\
&\qquad- \frac{1}{\left( b + 2 + \frac{2 \mathi (\tau + \mathi \epsilon)}{\beta} \right)} \hypergeom{2}{1}\left( 2, b + 2 + \frac{2 \mathi (\tau + \mathi \epsilon)}{\beta}; b + 3 + \frac{2 \mathi (\tau + \mathi \epsilon)}{\beta}; z \right) \Biggr]_{b = 0} \eqend{.}
\end{splitequation}
Using the expansion~\cite[\Eqs~(15.8.10) and~(15.8.12)]{dlmf} of the hypergeometric function near $z = 1$ and the recurrence relation~\cite[\Eq~(5.5.2)]{dlmf} of the digamma function $\psi$, the divergent terms cancel (as they must since the original integral was convergent) and we are left with
\begin{splitequation}
\partial_\tau I_0(\tau) &= - \frac{2 \mathi}{\beta^2} \lim_{\epsilon \to 0^+} \frac{\partial}{\partial b} \Biggl[ \left( b - 1 - \frac{2 \mathi (\tau + \mathi \epsilon)}{\beta} \right) \psi\left( b - \frac{2 \mathi (\tau + \mathi \epsilon)}{\beta} \right) \\
&\qquad- \left( b + 1 + \frac{2 \mathi (\tau + \mathi \epsilon)}{\beta} \right) \psi\left( b + 1 + \frac{2 \mathi (\tau + \mathi \epsilon)}{\beta} \right) \Biggr]_{b = 0} \eqend{.}
\end{splitequation}
Finally, performing the $b$ derivatives and using the reflection relation~\cite[Eq.~(5.5.4)]{dlmf} of the digamma function $\psi$, we obtain
\begin{equation}
\partial_\tau I_0(\tau) = - \frac{2 \pi}{\beta^2} \lim_{\epsilon \to 0^+} \Biggl[ \coth\left( \frac{2 \pi (\tau + \mathi \epsilon)}{\beta} \right) + \frac{\mathi \pi}{\beta} \bigl[ \beta + 2 \mathi (\tau + \mathi \epsilon) \bigr] \sin^{-2}\left( \frac{2 \pi (\tau + \mathi \epsilon)}{\beta} \right) \Biggr] \eqend{,}
\end{equation}
which is equal to the stated result.
\end{proof}
Using these results, a short computation shows that the relations
\begin{equations}[eq:appendix_smallp_relations]
\partial_\eta^2 \tilde{L}_\text{C}^{(1)}(\eta-\eta') &= - \frac{2}{3} \tilde{L}_\text{B}^{(0)}(\eta-\eta') \eqend{,} \\
\partial_\eta^2 \tilde{L}_\text{B}^{(1)}(\eta-\eta') &= \frac{1}{3} \tilde{L}_\text{B}^{(0)}(\eta-\eta') - \frac{4}{3} \tilde{L}_\text{A}^{(0)}(\eta-\eta') \eqend{,} \\
\partial_\eta^2 \tilde{L}_\text{C}^{(2)}(\eta-\eta') &= \frac{1}{5} \tilde{L}_\text{C}^{(1)}(\eta-\eta') - \frac{8}{5} \tilde{L}_\text{B}^{(1)}(\eta-\eta')
\end{equations}
hold between the coefficients~\eqref{eq:appendix_smallp_labc_components}. Inserting the expansions~\eqref{eq:appendix_smallp_labc} into the expressions~\eqref{eq:appendix_kernel_components} for the components of the kernel $\tilde{K}_{\mu\nu\rho\sigma}$ and employing the relations~\eqref{eq:appendix_smallp_relations}, we obtain
\begin{equations}[eq:appendix_smallp_kernel_components]
\tilde{K}^\beta_{0000}(\eta-\eta',\vec{p}) &= \tilde{L}_\text{A}^{(0)}(\eta-\eta') + \bigo{\vec{p}^2} \eqend{,} \\
\tilde{K}^\beta_{000i}(\eta-\eta',\vec{p}) &= - \frac{\mathi}{2} \vec{p}_i \partial_\eta \tilde{L}_\text{B}^{(1)}(\eta-\eta') + \bigo{\abs{\vec{p}}^3} \eqend{,} \\
\tilde{K}^\beta_{0i0j}(\eta-\eta',\vec{p}) &= - \frac{1}{3} \eta_{ij} \left[ \tilde{L}_\text{A}^{(0)}(\eta-\eta') - \tilde{L}_\text{B}^{(0)}(\eta-\eta') \right] + \bigo{\vec{p}^2} \eqend{,} \\
\tilde{K}^\beta_{00ij}(\eta-\eta',\vec{p}) &= \frac{1}{3} \eta_{ij} \tilde{L}_\text{A}^{(0)}(\eta-\eta') + \bigo{\vec{p}^2} \eqend{,} \\
\begin{split}
\tilde{K}^\beta_{0ikl}(\eta-\eta',\vec{p}) &= \frac{\mathi \eta_{i(j} \vec{p}_{k)}}{10} \partial_\eta \left[ 4 \tilde{L}_\text{B}^{(1)}(\eta-\eta') - 3 \tilde{L}_\text{C}^{(1)}(\eta-\eta') \right] \\
&\quad- \frac{\mathi \eta_{kl} \vec{p}_i}{10} \partial_\eta \left[ 3 \tilde{L}_\text{B}^{(1)}(\eta-\eta') - \tilde{L}_\text{C}^{(1)}(\eta-\eta') \right] + \bigo{\abs{\vec{p}}^3}
\end{split} \\
\tilde{K}^\beta_{ijkl}(\eta-\eta',\vec{p}) &= \frac{\eta_{ij} \eta_{kl} + 2 \eta_{i(k} \eta_{l)j}}{15} \tilde{L}_\text{A}^{(0)}(\eta-\eta') + \bigo{\vec{p}^2} \eqend{.}
\end{equations}
Finally, inserting the results for the coefficients~\eqref{eq:appendix_smallp_labc_components} and employing lemmas~\ref{lemma:integral_jmn} and~\ref{lemma:integral_in}, in particular \Eqs~\eqref{eq:appendix_smallp_j1424} and~\eqref{eq:appendix_smallp_dki0}, we obtain the stated results~\eqref{eq:kernel-thermal-smallp} for the small-momentum limit.

\subsection{Large temperature limit}
\label{app:larget}

The limit of large temperature $\beta \to 0^+$ is more difficult, since the convergence of the integrals is not uniform in $\beta$. To obtain a suitable expansion, we first rescale the integration variable $\vec{q} \to \vec{q}/\beta$ and expand the integrands to obtain
\begin{equations}[eq:appendix_larget_kabc]
\tilde{K}^\beta_{0000}(\tau,\vec{p}) &= \frac{1}{4 \beta^5} \left[ \hat{I}_{2,0}(\tau,\vec{p}) + \hat{J}_{2,0}(\tau,\vec{p}) \right] - \frac{\abs{\vec{p}}}{4 \beta^4} \left[ \hat{I}_{1,1}(\tau,\vec{p}) + \hat{J}_{1,1}(\tau,\vec{p}) \right] + \bigo{\beta^{-3}} \eqend{,} \\
\eta^{ij} \tilde{K}^\beta_{0i0j}(\tau,\vec{p}) &= - \frac{1}{4 \beta^5} \left[ \hat{I}_{2,0}(\tau,\vec{p}) - \hat{J}_{2,0}(\tau,\vec{p}) \right] + \frac{\abs{\vec{p}}}{4 \beta^4} \left[ \hat{I}_{1,1}(\tau,\vec{p}) - \hat{J}_{1,1}(\tau,\vec{p}) \right] + \bigo{\beta^{-3}} \eqend{,} \\
\eta^{ik} \eta^{jl} \tilde{K}^\beta_{ijkl}(\tau,\vec{p}) &= \tilde{K}_{0000}(\tau,\vec{p}) + \bigo{\beta^{-3}}
\end{equations}
from \Eqs~\eqref{eq:appendix_ka}, \eqref{eq:appendix_kb} and~\eqref{eq:appendix_kc}, where we defined
\begin{equation}
\label{eq:appendix_hatimn_def}
\hat{I}_{m,n}(\tau,\vec{p}) \doteq \lim_{\epsilon \to 0^+} \int \left[ \frac{\mathe^{\mathi \left( \abs{\beta \vec{p}-\vec{q}} + \abs{\vec{q}} \right) (T + \mathi \epsilon)}}{\left[ 1 - \mathe^{- \abs{\beta \vec{p}-\vec{q}}} \right] \left[ 1 - \mathe^{- \abs{\vec{q}}} \right]} + \frac{\mathe^{- \mathi \left( \abs{\beta \vec{p}-\vec{q}} + \abs{\vec{q}} \right) T}}{\left[ \mathe^{\abs{\beta \vec{p}-\vec{q}}} - 1 \right] \left[ \mathe^{\abs{\vec{q}}} - 1 \right]} \right] \abs{\vec{q}}^m \cos^n \theta \frac{\total^3 \vec{q}}{(2 \pi)^3}
\end{equation}
and
\begin{equation}
\label{eq:appendix_hatjmn_def}
\hat{J}_{m,n}(\tau,\vec{p}) \doteq \int \left[ \frac{\mathe^{- \mathi \left( \abs{\beta \vec{p}-\vec{q}} - \abs{\vec{q}} \right) T}}{\left[ \mathe^{\abs{\beta \vec{p}-\vec{q}}} - 1 \right] \left[ 1 - \mathe^{- \abs{\vec{q}}} \right]} + \frac{\mathe^{\mathi \left( \abs{\beta \vec{p}-\vec{q}} - \abs{\vec{q}} \right) T}}{\left[ 1 - \mathe^{- \abs{\beta \vec{p}-\vec{q}}} \right] \left[ \mathe^{\abs{\vec{q}}} - 1 \right]} \right] \abs{\vec{q}}^m \cos^n \theta \frac{\total^3 \vec{q}}{(2 \pi)^3} \eqend{,}
\end{equation}
and set $T \doteq \tau/\beta$. (Note that since after the rescaling the integrals are convergent for large $\abs{\vec{q}}$ for any $\beta$, we could take the limit $\epsilon \to 0$ inside the integral almost everywhere.) Expanding for fixed $T$, we see that we recover the small-momentum limit. If we keep instead $\tau$ fixed, we expect the exponentials in $\hat{I}_{m,n}$ to oscillate and give a vanishing result, while in $\hat{J}_{m,n}$ the oscillations cancel and we obtain a finite result.

That these expectation hold is seen as follows:
\begin{lemma}
\label{lemma:integral_hatjmn}
The integrals $\hat{J}_{m,n}$~\eqref{eq:appendix_hatjmn_def} have the following expansion for small $\beta$:
\begin{equations}
\hat{J}_{m,2n}(\tau,\vec{p}) &= \frac{(-1)^n}{\pi^2 \abs{\vec{p}}^{2n}} \left( 1 - \frac{\mathi \beta}{2} \frac{\partial}{\partial \tau} \right) \frac{\partial^{2n}}{\partial \tau^{2n}} \left[ \frac{\sin\left( \abs{\vec{p}} \tau \right)}{\abs{\vec{p}} \tau} \right] \Gamma(m+3) \zeta(m+2) + \bigo{\beta^2} \eqend{,} \\
\begin{split}
\hat{J}_{m,2n+1}(\tau,\vec{p}) &= - \frac{\mathi \beta}{2 \pi^2 (\mathi \abs{\vec{p}})^{2n+1}} \left[ (m+2) \frac{\partial}{\partial \tau} + \tau \left( \frac{\partial^2}{\partial \tau^2} + \vec{p}^2 \right) \right] \\
&\qquad\times \frac{\partial^{2n+1}}{\partial \tau^{2n+1}} \left[ \frac{\sin\left( \abs{\vec{p}} \tau \right)}{\abs{\vec{p}} \tau} \right] \Gamma(m+2) \zeta(m+1) + \bigo{\beta^2} \eqend{.}
\end{split}
\end{equations}
\end{lemma}
\begin{proof}
Expanding the integrand in $\beta$, writing the powers of $\cos \theta$ as derivatives wih respect to $\tau$ and performing the easy integrals over $\theta$ and $\phi$, we obtain
\begin{splitequation}
\hat{J}_{m,n}(\tau,\vec{p}) &= \frac{1 + (-1)^n}{2 \pi^2 (\mathi \abs{\vec{p}})^n} \frac{\partial^n}{\partial \tau^n} \left[ \frac{\sin\left( \abs{\vec{p}} \tau \right)}{\abs{\vec{p}} \tau} \right] \int_0^\infty \frac{q^{m+2}}{\left[ \mathe^q - 1 \right] \left[ 1 - \mathe^{- q} \right]} \total q \\
&\quad+ \frac{\mathi \beta}{2 \pi^2 (\mathi \abs{\vec{p}})^n} \frac{\partial^{n+1}}{\partial \tau^{n+1}} \left[ \frac{\sin\left( \abs{\vec{p}} \tau \right)}{\abs{\vec{p}} \tau} \right] \int_0^\infty \frac{\left[ (-1)^n \left( 1 - \mathe^{- q} \right) + 1 - \mathe^q \right] q^{m+2}}{\left[ \mathe^q - 1 \right]^2 \left[ 1 - \mathe^{- q} \right]^2} \total q \\
&\quad+ \frac{\mathi \beta \tau [ (-1)^n - 1 ]}{4 \pi^2 (\mathi \abs{\vec{p}})^n} \frac{\partial^n}{\partial \tau^n} \left( \frac{\partial^2}{\partial \tau^2} + \vec{p}^2 \right) \left[ \frac{\sin\left( \abs{\vec{p}} \tau \right)}{\abs{\vec{p}} \tau} \right] \int_0^\infty \frac{q^{m+1}}{\left[ \mathe^q - 1 \right] \left[ 1 - \mathe^{- q} \right]} \total q \\
&\quad+ \bigo{\beta^2} \eqend{.}
\end{splitequation}
Performing a partial fraction decomposition of the integrand of the remaining integrals and using the integral representations~\cite[\Eqs~(25.5.1) and~(25.5.2)]{dlmf} of the Riemann $\zeta$ function, we obtain the stated result.
\end{proof}
In particular, we have $\hat{J}_{1,1} = \bigo{\beta}$ and
\begin{equation}
\label{eq:appendix_larget_hatj20}
\hat{J}_{2,0}(\tau,\vec{p}) = \frac{4 \pi^2}{15} \frac{\sin\left( \abs{\vec{p}} \tau \right)}{\abs{\vec{p}} \tau} - \frac{2 \mathi \beta \pi^2}{15} \frac{\abs{\vec{p}} \tau \cos\left( \abs{\vec{p}} \tau \right) - \sin\left( \abs{\vec{p}} \tau \right)}{\abs{\vec{p}} \tau^2} + \bigo{\beta^2} \eqend{.}
\end{equation}

For the oscillatory integrals, we have instead
\begin{lemma}
The integrals $\hat{I}_{m,n}$~\eqref{eq:appendix_hatimn_def} fulfill
\begin{equation}
\hat{I}_{m,n}(\tau,\vec{p}) = \bigo{\beta^{m+1}} \eqend{.}
\end{equation}
\end{lemma}
\begin{proof}
The leading term of $\hat{I}_{m,n}$ for small $\beta$ is given by
\begin{splitequation}
\hat{I}_{m,n}(\tau,\vec{p}) &\sim \frac{1 + (-1)^n}{4 \pi^2 (n+1)} \lim_{\epsilon \to 0^+} \int_0^\infty \left[ \frac{\mathe^{2 \mathi q (T + \mathi \epsilon)}}{\left( 1 - \mathe^{- q} \right)^2} + \frac{\mathe^{- 2 \mathi q T}}{\left( \mathe^q - 1 \right)^2} \right] q^{m+2} \total q \\
&= \frac{1 + (-1)^n}{2^{m+2} \pi^2 (n+1)} \lim_{\epsilon \to 0^+} \frac{\partial^m}{\partial T^m} \int_0^\infty \left[ \frac{\mathe^{2 \mathi q (T + \mathi \epsilon)}}{\left( 1 - \mathe^{- q} \right)^2} + \frac{(-1)^m \mathe^{- 2 \mathi q T}}{\left( \mathe^q - 1 \right)^2} \right] q^2 \total q \eqend{,}
\end{splitequation}
where we used that
\begin{equation}
\int_0^\pi \cos^n \theta \sin \theta \total \theta = \frac{1 + (-1)^n}{n+1} \eqend{.}
\end{equation}
Changing variables to $q = - \ln z$, this reduces to
\begin{equation}
\hat{I}_{m,n}(\tau,\vec{p}) = \frac{1 + (-1)^n}{2^{m+2} \pi^2 (n+1)} \lim_{\epsilon \to 0^+} \frac{\partial^m}{\partial T^m} \int_0^1 \left[ z^{- 1 - 2 \mathi (T + \mathi \epsilon)} + (-1)^m z^{1 + 2 \mathi T} \right] \frac{\ln^2 z}{(1-z)^2} \total z \eqend{,}
\end{equation}
and we note that for $\epsilon > 0$, the integrand is integrable both at $z = 0$ and $z = 1$ and the integral thus well-defined. We employ again the integral representation~\eqref{eq:appendix_smallp_hypergeom_intlog} of the hypergeometric function (with $a = 2$ and $c = -1$) and obtain after some simplifications
\begin{splitequation}
\hat{I}_{m,n}(\tau,\vec{p}) &\sim - \frac{1 + (-1)^n}{2^{m+5} \pi^2 (n+1)} \lim_{\epsilon \to 0^+} \frac{\partial^{m+2}}{\partial T^{m+2}} \lim_{z \to 1^-} \biggl[ (-1)^m \frac{1}{1 + \mathi T} \hypergeom{2}{1}\left( 2, 2 + 2 \mathi T; 3 + 2 \mathi T; z \right) \\
&\hspace{4em}+ \frac{\mathi}{T + \mathi \epsilon} \hypergeom{2}{1}\left( 2, - 2 \mathi (T + \mathi \epsilon); 1 - 2 \mathi (T + \mathi \epsilon); z \right) \biggr] \\
&= \frac{1 + (-1)^n}{2^{m+4} \pi^2 (n+1)} \beta^{m+2} \lim_{\epsilon \to 0^+} \frac{\partial^{m+2}}{\partial \tau^{m+2}} \\
&\qquad\times \biggl[ \left( 1 + \frac{2 \mathi \tau}{\beta} \right) \left[ \psi\left( - \frac{2 \mathi \tau}{\beta} + 2 \epsilon \right) - (-1)^m \psi\left( 1 + \frac{2 \mathi \tau}{\beta} \right) \right] \biggr] \eqend{,}
\end{splitequation}
where in the second equality we employed the expansion~\cite[\Eqs~(15.8.10) and~(15.8.12)]{dlmf} of the hypergeometric function near $z = 1$ and the recurrence relation~\cite[\Eq~(5.5.2)]{dlmf} of the digamma function $\psi$. The expansion of the digamma function for large argument is also known~\cite[\Eq~(5.11.2)]{dlmf}, and we insert it into our result and perform the derivatives using the general Leibniz rule. It follows that the leading behaviour of $\hat{I}_{m,n}(\tau,\vec{p})$ for small $\beta$ is given by
\begin{equation}
\hat{I}_{m,n}(\tau,\vec{p}) \sim \mathi \frac{[ 1 + (-1)^n ] [ 1 - (-1)^m ]}{2^{m+3} \pi^2 (n+1)} \beta^{m+1} \frac{\partial^{m+1}}{\partial \tau^{m+1}} \ln\left( \mathi \tau \right) \eqend{,}
\end{equation}
which proves the assertion.
\end{proof}

Employing the above results, in particular \Eq~\eqref{eq:appendix_larget_hatj20}, the expansion~\eqref{eq:appendix_larget_kabc} of the kernels for large temperature is given by
\begin{equations}[eq:appendix_larget_kabc_results]
\tilde{K}^\beta_{0000}(\tau,\vec{p}) &= \frac{\pi^2}{15 \beta^5} \frac{\sin\left( \abs{\vec{p}} \tau \right)}{\abs{\vec{p}} \tau} - \frac{\mathi \pi^2}{30 \beta^4} \frac{\abs{\vec{p}} \tau \cos\left( \abs{\vec{p}} \tau \right) - \sin\left( \abs{\vec{p}} \tau \right)}{\abs{\vec{p}} \tau^2} + \bigo{\beta^{-3}} \eqend{,} \\
\eta^{ij} \tilde{K}^\beta_{0i0j}(\tau,\vec{p}) &= \tilde{K}^\beta_{0000}(\tau,\vec{p}) + \bigo{\beta^{-3}} \eqend{,} \\
\eta^{ik} \eta^{jl} \tilde{K}^\beta_{ijkl}(\tau,\vec{p}) &= \tilde{K}^\beta_{0000}(\tau,\vec{p}) + \bigo{\beta^{-3}} \eqend{.}
\end{equations}
It follows that $\tilde{L}_\text{C}(\tau,\vec{p}) = 4 \tilde{L}_\text{A}(\tau,\vec{p}) + \bigo{\beta^{-3}}$, $\tilde{L}_\text{B}(\tau,\vec{p}) = 2 \tilde{L}_\text{A}(\tau,\vec{p}) + \bigo{\beta^{-3}}$, and inserting these into the components~\eqref{eq:appendix_kernel_components} of the noise kernel $\tilde{K}^\beta_{\mu\nu\rho\sigma}$, we obtain the stated results~\eqref{eq:kernel-thermal-larget} for the large temperature limit.

\subsection{Small temperature limit}
\label{app:smallt}

The last limit that we consider is the one of small temperatures $\beta \to \infty$. Here we have to separate the zero-temperature contribution from the temperature-dependent ones to obtain
\begin{splitequation}
\label{eq:appendix_smallt_ka}
\tilde{K}^\beta_{0000}(\abs{\vec{p}},\tau) &= \frac{1}{4} \tilde{I}_{1,1}(\tau,\vec{p}) \\
&\quad+ \int \abs{\vec{p}-\vec{q}} \abs{\vec{q}} \frac{\mathe^{\mathi \abs{\vec{p}-\vec{q}} \tau} \cos\left( \abs{\vec{q}} \tau \right)}{\mathe^{\beta \abs{\vec{q}}} - 1} \frac{\total^3 \vec{q}}{(2 \pi)^3} \\
&\quad+ \int \abs{\vec{p}-\vec{q}} \abs{\vec{q}} \frac{\cos\left( \abs{\vec{p}-\vec{q}} \tau \right) \cos\left( \abs{\vec{q}} \tau \right)}{\left( \mathe^{\beta \abs{\vec{p}-\vec{q}}} - 1 \right) \left( \mathe^{\beta \abs{\vec{q}}} - 1 \right)} \frac{\total^3 \vec{q}}{(2 \pi)^3} \eqend{,}
\end{splitequation}
\begin{splitequation}
\label{eq:appendix_smallt_kb}
\eta^{ij} \tilde{K}^\beta_{0i0j}(\tau,\vec{p}) &= \frac{1}{8} \vec{p}^2 \tilde{I}_{0,0}(\tau,\vec{p}) - \frac{1}{4} \tilde{I}_{0,2}(\tau,\vec{p}) \\
&\quad+ \mathi \int \frac{\vec{p}^2 - \vec{q}^2 - \abs{\vec{p}-\vec{q}}^2}{2} \frac{\mathe^{\mathi \abs{\vec{p}-\vec{q}} \tau} \sin\left( \abs{\vec{q}} \tau \right)}{\mathe^{\beta \abs{\vec{q}}} - 1} \frac{\total^3 \vec{q}}{(2 \pi)^3} \\
&\quad- \int \frac{\vec{p}^2 - \vec{q}^2 - \abs{\vec{p}-\vec{q}}^2}{2} \frac{\sin\left( \abs{\vec{p}-\vec{q}} \tau \right) \sin\left( \abs{\vec{q}} \tau \right)}{\left( \mathe^{\beta \abs{\vec{p}-\vec{q}}} - 1 \right) \left( \mathe^{\beta \abs{\vec{q}}} - 1 \right)} \frac{\total^3 \vec{q}}{(2 \pi)^3} \eqend{,}
\end{splitequation}
\begin{splitequation}
\label{eq:appendix_smallt_kc}
\eta^{ik} \eta^{jl} \tilde{K}^\beta_{ijkl}(\tau,\vec{p}) &= \frac{1}{16} \abs{\vec{p}}^4 \tilde{I}_{-1,-1}(\tau,\vec{p}) - \frac{1}{4} \vec{p}^2 \tilde{I}_{-1,1}(\tau,\vec{p}) + \frac{1}{8} \tilde{I}_{-1,3}(\tau,\vec{p}) + \frac{1}{8} \tilde{I}_{1,1}(\tau,\vec{p}) \\
&\quad+ \int \frac{\left( \vec{p}^2 - \vec{q}^2 - \abs{\vec{p}-\vec{q}}^2 \right)^2}{4 \abs{\vec{p}-\vec{q}} \abs{\vec{q}}} \frac{\mathe^{\mathi \abs{\vec{p}-\vec{q}} \tau} \cos\left( \abs{\vec{q}} \tau \right)}{\mathe^{\beta \abs{\vec{q}}} - 1} \frac{\total^3 \vec{q}}{(2 \pi)^3} \\
&\quad+ \int \frac{\left( \vec{p}^2 - \vec{q}^2 - \abs{\vec{p}-\vec{q}}^2 \right)^2}{4 \abs{\vec{p}-\vec{q}} \abs{\vec{q}}} \frac{\cos\left( \abs{\vec{p}-\vec{q}} \tau \right) \cos\left( \abs{\vec{q}} \tau \right)}{\left( \mathe^{\beta \abs{\vec{p}-\vec{q}}} - 1 \right) \left( \mathe^{\beta \abs{\vec{q}}} - 1 \right)} \frac{\total^3 \vec{q}}{(2 \pi)^3} \eqend{,}
\end{splitequation}
from the expressions~\eqref{eq:appendix_ka}, \eqref{eq:appendix_kb} and~\eqref{eq:appendix_kc}, where we defined
\begin{equation}
\label{eq:appendix_tildeimn_def}
\tilde{I}_{m,n}(\tau,\vec{p}) \doteq \lim_{\epsilon \to 0^+} \int \mathe^{\mathi (\abs{\vec{p}-\vec{q}}+\abs{\vec{q}}) (\tau+\mathi \epsilon)} \abs{\vec{p}-\vec{q}}^m \abs{\vec{q}}^n \frac{\total^3 \vec{q}}{(2 \pi)^3} \eqend{.}
\end{equation}
Note that we could interchange the limit $\epsilon \to 0$ with the integration in the temperature-dependent terms since they converge at infinity, and have changed the integration variable $\vec{q} \to \vec{p}-\vec{q}$ in some terms to simplify the result.

For the temperature-independent contributions, we obtain
\begin{lemma}
The integrals $\tilde{I}_{m,n}$~\eqref{eq:appendix_tildeimn_def} with $m \geq 0$ and $n \geq -2$ satisfy the recursion relation
\begin{equation}
\tilde{I}_{m,n}(\tau,\vec{p}) = - \mathi \frac{\partial}{\partial \tau} \tilde{I}_{m-1,n}(\tau,\vec{p}) - \tilde{I}_{m-1,n+1}(\tau,\vec{p}) \eqend{.}
\end{equation}
For $m = - 1$ and $n \geq -2$, we have
\begin{splitequation}
\tilde{I}_{-1,n}(\tau,\vec{p}) &= \frac{(n+1)! \mathi^{n+2}}{2^{n+3} \pi^2} \left[ \frac{\sin\left( \abs{\vec{p}} \tau \right)}{\abs{\vec{p}} \tau} \lim_{\epsilon \to 0^+} \frac{1}{\left( \tau + \mathi \epsilon \right)^{n+2}} - \frac{\mathi \mathe^{\mathi \abs{\vec{p}} \tau}}{2 \abs{\vec{p}} \tau^{n+3}} \mathbb{T}_{n+3} \mathe^{- 2 \mathi \abs{\vec{p}} \tau} \right] \eqend{,}
\end{splitequation}
where we defined the operator $\mathbb{T}_n$ by
\begin{equation}
\label{eq:appendix_smallt_taylor}
\mathbb{T}_n f(\tau) \doteq f(\tau) - \sum_{k=0}^{n-1} \frac{\tau^k}{k!} f^{(k)}(0) \eqend{,}
\end{equation}
namely subtracting the first $n$ terms of the Taylor series of $f$. Note that $\mathbb{T}_n$ commutes with derivatives and multiplication by powers of $\tau$ in the sense that
\begin{equation}
\label{eq:appendix_smallt_taylor_d}
\partial_\tau \mathbb{T}_n f(\tau) = \mathbb{T}_{n-1} f'(\tau) \eqend{,} \quad \mathbb{T}_n \left[ \tau^k f(\tau) \right] = \tau^k \mathbb{T}_{n-k} f(\tau) \eqend{,}
\end{equation}
which both follow from the generalized Leibniz rule for derivatives.
\end{lemma}
\begin{proof}
The first equality follows by exchanging the $\tau$ derivative with the integration, which is allowed under the stated assumptions. For the second result, we use that
\begin{equation}
\frac{1}{\abs{\vec{p}-\vec{q}}} \mathe^{\mathi \abs{\vec{p}-\vec{q}} \tau} \sin \theta = - \frac{\mathi}{\abs{\vec{p}} \abs{\vec{q}} \tau} \frac{\partial}{\partial \theta} \mathe^{\mathi \sqrt{\vec{p}^2 + \vec{q}^2 - 2 \abs{\vec{p}} \abs{\vec{q}} \cos \theta} \, \tau}
\end{equation}
to perform the integral over $\theta$ and obtain
\begin{splitequation}
\tilde{I}_{-1,n}(\tau,\vec{p}) &= - \frac{\mathi}{4 \pi^2 \abs{\vec{p}} \tau} \lim_{\epsilon \to 0^+} \int_0^\infty \left[ \mathe^{\mathi (\abs{\vec{p}} + 2 q) (\tau + \mathi \epsilon)} - \mathe^{\mathi (q + \abs{\abs{\vec{p}} - q}) (\tau + \mathi \epsilon)} \right] q^{n+1} \total q \\
&= \frac{\sin\left( \abs{\vec{p}} \tau \right)}{2 \pi^2 \abs{\vec{p}} \tau} \lim_{\epsilon \to 0^+} \int_0^\infty \mathe^{2 \mathi q (\tau + \mathi \epsilon)} q^{n+1} \total q + \frac{1}{2 \pi^2 \abs{\vec{p}} \tau} \int_0^\abs{\vec{p}} \mathe^{\mathi q \tau} \sin\left[ \left( q - \abs{\vec{p}} \right) \tau \right] q^{n+1} \total q \eqend{.}
\raisetag{3em}
\end{splitequation}
For the first integral, we employ the integral representation~\cite[\Eq~(5.9.1)]{dlmf} of the $\Gamma$ function, while for the second one we change variables $q = \abs{\vec{p}} t$ to obtain
\begin{equation}
\tilde{I}_{-1,n}(\tau,\vec{p}) = \frac{\sin\left( \abs{\vec{p}} \tau \right)}{2^{n+3} \pi^2 \abs{\vec{p}} \tau} \lim_{\epsilon \to 0^+} \frac{\Gamma(n+2)}{\left( \epsilon - \mathi \tau \right)^{n+2}} + \frac{\mathi \abs{\vec{p}}^{n+1}}{4 \pi^2 \tau} \mathe^{\mathi \abs{\vec{p}} \tau} \int_0^1 \left[ 1 - \mathe^{- 2 \mathi \abs{\vec{p}} \tau (1-t)} \right] t^{n+1} \total t \eqend{.}
\end{equation}
Employing the integral representation~\cite[Eq.~(8.2.7)]{dlmf} of the incomplete $\gamma$ function, it follows that
\begin{splitequation}
\tilde{I}_{-1,n}(\tau,\vec{p}) &= \frac{\sin\left( \abs{\vec{p}} \tau \right)}{2^{n+3} \pi^2 \abs{\vec{p}} \tau} \mathi^{n+2} \lim_{\epsilon \to 0^+} \frac{\Gamma(n+2)}{\left( \tau + \mathi \epsilon \right)^{n+2}} \\
&\quad+ \frac{\mathi \abs{\vec{p}}^{n+1}}{4 \pi^2 \tau} \left[ \frac{\mathe^{\mathi \abs{\vec{p}} \tau}}{n+2} - \mathe^{- \mathi \abs{\vec{p}} \tau} \frac{\gamma(n+2,- 2 \mathi \abs{\vec{p}} \tau)}{(-2 \mathi \abs{\vec{p}} \tau)^{n+2}} \right] \eqend{.}
\raisetag{2.6em}
\end{splitequation}
Finally, using the values~\cite[\Eqs~(8.4.7) and (8.4.11)]{dlmf} of the incomplete $\gamma$ function, the second equality is proven.
\end{proof}
For the integrals appearing in the kernels~\eqref{eq:appendix_smallt_ka}, \eqref{eq:appendix_smallt_kb} and~\eqref{eq:appendix_smallt_kc}, we obtain using these results
\begin{equations}[eq:appendix_smallt_tildei]
\tilde{I}_{-1,-1}(\tau,\vec{p}) &= \frac{\mathi \mathe^{\mathi \abs{\vec{p}} \tau}}{4 \pi^2} \lim_{\epsilon \to 0^+} \frac{1}{\tau + \mathi \epsilon} \eqend{,} \\
\tilde{I}_{-1,1}(\tau,\vec{p}) &= - \frac{\mathi \mathe^{\mathi \abs{\vec{p}} \tau} \left( 3 - 3 \mathi \abs{\vec{p}} \tau - 2 \vec{p}^2 \tau^2 \right)}{24 \pi^2} \lim_{\epsilon \to 0^+} \frac{1}{\left( \tau + \mathi \epsilon \right)^3} \eqend{,} \\
\tilde{I}_{-1,3}(\tau,\vec{p}) &= \frac{\mathi \mathe^{\mathi \abs{\vec{p}} \tau} \left( 15 - 15 \mathi \abs{\vec{p}} \tau - 10 \vec{p}^2 \tau^2 + 5 \mathi \abs{\vec{p}}^3 \tau^3 + 2 \abs{\vec{p}}^4 \tau^4 \right)}{40 \pi^2} \lim_{\epsilon \to 0^+} \frac{1}{\left( \tau + \mathi \epsilon \right)^5} \eqend{,} \\
\tilde{I}_{0,0}(\tau,\vec{p}) &= - \frac{\mathi \mathe^{\mathi \abs{\vec{p}} \tau} \left( 3 - 3 \mathi \abs{\vec{p}} \tau - \vec{p}^2 \tau^2 \right)}{24 \pi^2} \lim_{\epsilon \to 0^+} \frac{1}{\left( \tau + \mathi \epsilon \right)^3} \eqend{,} \\
\tilde{I}_{0,2}(\tau,\vec{p}) &= \frac{\mathi \mathe^{\mathi \abs{\vec{p}} \tau} \left( 30 - 30 \mathi \abs{\vec{p}} \tau - 15 \vec{p}^2 \tau^2 + 5 \mathi \abs{\vec{p}}^3 \tau^3 + \abs{\vec{p}}^4 \tau^4 \right)}{80 \pi^2} \lim_{\epsilon \to 0^+} \frac{1}{\left( \tau + \mathi \epsilon \right)^5} \eqend{,} \\
\tilde{I}_{1,1}(\tau,\vec{p}) &= \frac{\mathi \mathe^{\mathi \abs{\vec{p}} \tau} \left( 45 - 45 \mathi \abs{\vec{p}} \tau - 20 \vec{p}^2 \tau^2 + 5 \mathi \abs{\vec{p}}^3 \tau^3 + \abs{\vec{p}}^4 \tau^4 \right)}{120 \pi^2} \lim_{\epsilon \to 0^+} \frac{1}{\left( \tau + \mathi \epsilon \right)^5} \eqend{.}
\end{equations}

For the temperature-dependent contributions, we cut the integration range into $\abs{\vec{q}} \leq \abs{\vec{p}}$ and the remainder. The remainder decays at least $\sim \mathe^{- \beta \abs{\vec{p}}}$, namely exponentially with large $\beta$, and we do not need to consider it in the expansion. In the first term, we rescale the integration variable $\vec{q} \to \vec{q}/\beta$ and then take the limit of large $\beta$. The remaining integrals can be done using the integral representations~\cite[\Eqs~(25.5.1) and~(25.5.2)]{dlmf} of the Riemann $\zeta$ function, and we obtain the results
\begin{equations}
\begin{split}
\tilde{K}^\beta_{0000}(\abs{\vec{p}},\tau) &= \frac{\mathi \mathe^{\mathi \abs{\vec{p}} \tau} \left( 45 - 45 \mathi \abs{\vec{p}} \tau - 20 \vec{p}^2 \tau^2 + 5 \mathi \abs{\vec{p}}^3 \tau^3 + \abs{\vec{p}}^4 \tau^4 \right)}{480 \pi^2} \lim_{\epsilon \to 0^+} \frac{1}{\left( \tau + \mathi \epsilon \right)^5} \\
&\quad+ \mathe^{\mathi \abs{\vec{p}} \tau} \abs{\vec{p}} \frac{\pi^2}{30 \beta^4} + \bigo{\beta^{-6}} \eqend{,} \raisetag{1.8em}
\end{split} \\
\eta^{ij} \tilde{K}^\beta_{0i0j}(\tau,\vec{p}) &= - \frac{\mathi \mathe^{\mathi \abs{\vec{p}} \tau} \left( 45 - 45 \mathi \abs{\vec{p}} \tau - 15 \vec{p}^2 \tau^2 - \abs{\vec{p}}^4 \tau^4 \right)}{480 \pi^2} \lim_{\epsilon \to 0^+} \frac{1}{\left( \tau + \mathi \epsilon \right)^5} + \bigo{\beta^{-6}} \eqend{,} \\
\begin{split}
\eta^{ik} \eta^{jl} \tilde{K}^\beta_{ijkl}(\tau,\vec{p}) &= \frac{\mathi \mathe^{\mathi \abs{\vec{p}} \tau} \left( 45 - 45 \mathi \abs{\vec{p}} \tau - 10 \vec{p}^2 \tau^2 - 5 \mathi \abs{\vec{p}}^3 \tau^3 + \abs{\vec{p}}^4 \tau^4 \right)}{480 \pi^2} \lim_{\epsilon \to 0^+} \frac{1}{\left( \tau + \mathi \epsilon \right)^5} \\
&\quad+ \mathe^{\mathi \abs{\vec{p}} \tau} \abs{\vec{p}} \frac{\pi^2}{90 \beta^4} + \bigo{\beta^{-6}} \raisetag{1.8em}
\end{split}
\end{equations}
for the kernels~\eqref{eq:appendix_smallt_ka}, \eqref{eq:appendix_smallt_kb} and~\eqref{eq:appendix_smallt_kc}. It follows that
\begin{equations}
\begin{split}
\tilde{L}_\text{A}(\abs{\vec{p}},\tau) &= \frac{\mathi \mathe^{\mathi \abs{\vec{p}} \tau} \left( 45 - 45 \mathi \abs{\vec{p}} \tau - 20 \vec{p}^2 \tau^2 + 5 \mathi \abs{\vec{p}}^3 \tau^3 + \abs{\vec{p}}^4 \tau^4 \right)}{480 \pi^2} \lim_{\epsilon \to 0^+} \frac{1}{\left( \tau + \mathi \epsilon \right)^5} \\
&\quad+ \mathe^{\mathi \abs{\vec{p}} \tau} \abs{\vec{p}} \frac{\pi^2}{30 \beta^4} + \bigo{\beta^{-6}} \eqend{,}
\end{split} \\
\tilde{L}_\text{B}(\tau,\vec{p}) &= - \frac{\mathi \vec{p}^2 \mathe^{\mathi \abs{\vec{p}} \tau} \left( 5 - 5 \mathi \abs{\vec{p}} \tau - 2 \vec{p}^2 \tau^2 \right)}{480 \pi^2} \lim_{\epsilon \to 0^+} \frac{1}{\left( \tau + \mathi \epsilon \right)^3} + \mathe^{\mathi \abs{\vec{p}} \tau} \abs{\vec{p}} \frac{\pi^2}{30 \beta^4} + \bigo{\beta^{-6}} \eqend{,} \\
\tilde{L}_\text{C}(\tau,\vec{p}) &= \frac{\mathi \abs{\vec{p}}^4 \mathe^{\mathi \abs{\vec{p}} \tau}}{120 \pi^2} \lim_{\epsilon \to 0^+} \frac{1}{\tau + \mathi \epsilon} + \mathe^{\mathi \abs{\vec{p}} \tau} \abs{\vec{p}} \frac{2 \pi^2}{45 \beta^4} + \bigo{\beta^{-6}} \eqend{.}
\end{equations}
Inserting these results into the components~\eqref{eq:appendix_kernel_components} of the noise kernel and employing the identity~\eqref{eq:appendix_smallt_taylor_d}, we obtain the stated results~\eqref{eq:kernel-thermal-smallt} for the small temperature limit.

\bibliography{bibliography}

\end{document}